\documentclass[11pt]{article}
\pdfoutput=1 %

\usepackage[utf8]{inputenc}

\usepackage[letterpaper,margin=1in]{geometry}
\usepackage[OT1]{fontenc}

\usepackage{amsmath,amssymb,amsthm,amsfonts,latexsym,bbm,xspace,thm-restate}
\usepackage{graphicx,float,mathtools,braket,slantsc}
\usepackage{enumitem,booktabs,forest,mathdots,soul}
\usepackage[useregional]{datetime2}
\DTMusemodule{english}{en-US}

\usepackage{qcircuit,tikz}

\usepackage[colorlinks,citecolor=blue,,linkcolor=magenta,bookmarks=true]{hyperref}
\usepackage[capitalise]{cleveref}

\usepackage{multicol,array}
\usepackage{lipsum,framed}

\newtheorem{theorem}{Theorem}
\newtheorem*{theorem*}{Theorem}
\newtheorem{lemma}[theorem]{Lemma}

\newtheorem{proposition}[theorem]{Proposition}
\newtheorem{corollary}[theorem]{Corollary}
\newtheorem{conjecture}[theorem]{Conjecture}

\crefname{result}{Result}{Results}

\crefname{fact}{Fact}{Facts}

\theoremstyle{definition}
\newtheorem{definition}[theorem]{Definition}

\newcommand{\ketbra}[2]{\ket{#1}\!\bra{#2}}

\newcommand{\unitary}[1]{\textup{#1}}
\newcommand{\cont}{\unitary{controlled-}}
\newcommand{\SWAP}{\unitary{SWAP}}

\newcommand{\NOT}{\unitary{NOT}}

\newcommand{\complex}{\mathbb C}
\newcommand{\nats}{\mathbb N}

\newcommand{\class}[1]{\ensuremath{\mathsf{#1}}\xspace}
\mathchardef\mhyphen="2D %

\newcommand{\BQP}{\class{BQP}}
\newcommand{\QCMA}{\class{QCMA}}
\newcommand{\QMA}{\class{QMA}}

\newcommand{\prob}[1]{\textup{\textsc{#1}}\xspace}

\newcommand{\permInvers}{\prob{Permutation Inversion}}
\newcommand{\functionErasure}{\prob{Function Erasure}}
\newcommand{\indexErasure}{\prob{Index Erasure}}
\newcommand{\graphIso}{\prob{Graph Isomorphism}}
\newcommand{\rigidGraphIso}{\prob{Rigid Graph Isomorphism}}
\newcommand{\setComp}{\prob{Set Comparison}}
\newcommand{\setEquality}{\prob{Set Equality}}
\newcommand{\perminvgarb}{\prob{Embedded PermInv}}
\newcommand{\collision}{\prob{Collision}}

\newcommand{\calA}{\mathcal{A}}
\newcommand{\calB}{\mathcal{B}}
\newcommand{\calC}{\mathcal{C}}
\newcommand{\calD}{\mathcal{D}}

\newcommand{\calO}{\mathcal{O}}

\newcommand{\regA}{\calA}
\newcommand{\regB}{\calB}
\newcommand{\regC}{\calC}
\newcommand{\regD}{\calD}

\DeclarePairedDelimiter\absd{\lvert}{\rvert}
\DeclarePairedDelimiter\normd{\lVert}{\rVert}
\DeclarePairedDelimiter\parend{\lparen}{\rparen}
\DeclarePairedDelimiter\floord{\lfloor}{\rfloor}

\newcommand{\abs}[1]{\absd*{#1}}
\newcommand{\norm}[1]{\normd*{#1}}
\newcommand{\paren}[1]{\parend*{#1}}
\newcommand{\floor}[1]{\floord*{#1}}

\newcommand{\bO}[1]{\operatorname*{O}\paren{#1}}

\newcommand{\lO}[1]{\operatorname*{o}\paren{#1}}
\newcommand{\bOm}[1]{\operatorname*{\Omega}\paren{#1}}

\newcommand{\bOnoparen}{\operatorname*{O}}
\newcommand{\bOmnoparen}{\operatorname*{\Omega}}
\newcommand{\bTnoparen}{\operatorname*{\Theta}}
\newcommand{\lonoparen}{\operatorname*{o}}

\newcommand{\ie}{i.e.\xspace}

\newcommand{\eg}{e.g.\xspace}

\newcommand{\bin}[0]{\left\{0,1\right\}}
\renewcommand{\set}[1]{\left\{#1\right\}}

\newcommand{\piinv}{\pi^{-1}}
\newcommand{\inv}{^{-1}}

\newcommand{\XOR}{\textsc{xor}}

\newcommand{\oracle}[1]{\textup{#1}}

\newcommand{\Ppi}{\oracle{P}\!_{\pi}}
\newcommand{\Ppiinv}{\oracle{P}\!_{\pi^{-1}}}
\newcommand{\Spi}{\oracle{S}_{\pi}}
\newcommand{\Spiinv}{\oracle{S}_{\pi^{-1}}}

\newcommand{\diffusion}{\ensuremath{\oracle{D}}}

\DeclareMathOperator{\relg}{\overrightarrow{\gamma_2}}
\newcommand{\zo}{\bin}

\newcommand{\p}[1]{\paren{#1}}

\newcommand{\cb}[1]{\set{#1}}
\newcommand{\ang}[1]{\left\langle#1\right\rangle}

\DeclareMathOperator{\g}{\gamma_2}

\newcommand{\xor}{\mathsf{xor}}
\newcommand{\phase}{\mathsf{phase}}
\newcommand{\perm}{\mathsf{perm}}
\newcommand{\dxor}{\Delta_{\xor}}
\newcommand{\dperm}{\Delta_{\perm}}
\newcommand{\dphase}{\Delta_{\phase}}
\DeclareMathOperator{\lmax}{\lambda_{\mathsf{max}}}

\DeclareMathOperator{\ql}{Q_{\mathsf{LV}}}
\DeclareMathOperator{\qbperm}{Q_{\mathsf{BE}}^{\perm}}
\DeclareMathOperator{\qlperm}{Q_{\mathsf{LV}}^{\perm}}
\DeclareMathOperator{\qbxor}{Q_{\mathsf{BE}}^{\xor}}
\DeclareMathOperator{\qlxor}{Q_{\mathsf{LV}}^{\xor}}

\DeclareMathOperator{\relgext}{\overrightarrow{\gamma_2}^{\mathsf{ext}}}
\DeclareMathOperator{\relgstd}{\overrightarrow{\gamma_2}^{\mathsf{std}}}

\title{
Quantum Search with In-Place Queries
}

\author{
    Blake Holman\\\small{\textsl{Sandia National Laboratories}}\\\small{\textsl{Purdue University}}\\\small{\texttt{\href{mailto:blake@perdue.edu}{holman14@purdue.edu}}}\and
    Ronak Ramachandran\\\small{\textsl{The University of Texas at Austin}}\\\small{\texttt{\href{mailto:ronakr@utexas.edu}{ronakr@utexas.edu}}}\and
    Justin Yirka\\\small{\textsl{Sandia National Laboratories}}\\\small{\textsl{The University of Texas at Austin}}\\\small{\texttt{\href{mailto:yirka@utexas.edu}{yirka@utexas.edu}}}
}

\hypersetup{
    pdftitle={Quantum Search with In-Place Queries},
    pdfauthor={Blake Holman, Ronak Ramachandran, Justin Yirka},
}

\date{April 2025}

\begin{document}

\maketitle

\begin{abstract}
    Quantum query complexity is typically characterized in terms of
    \XOR{} queries $\ket{x,y}\mapsto \ket{x,y\oplus f(x)}$
    or phase queries,
    which ensure that even queries to non-invertible functions are unitary.
    When querying a permutation, another natural model is unitary:
    in-place queries $\ket{x}\mapsto \ket{f(x)}$.

    Some problems are known to require exponentially fewer in-place queries than \XOR{} queries,
    but no separation has been shown in the opposite direction.
    A candidate for such a separation was the problem of inverting a permutation over $N$ elements.
    This task, equivalent to unstructured search in the context of permutations,
    is solvable with $\bOnoparen(\sqrt{N})$ \XOR{} queries but was conjectured to require
    $\bOm{N}$ in-place queries.

    We refute this conjecture by designing a quantum algorithm for Permutation Inversion using
    $\bOnoparen(\sqrt{N})$ in-place queries.
    Our algorithm achieves the same speedup as Grover's algorithm
    despite the inability to efficiently uncompute queries
    or perform straightforward oracle-controlled reflections.

    Nonetheless, we show that there are indeed problems which require fewer \XOR{} queries than in-place queries.
    We introduce a subspace-conversion problem called Function Erasure
    that requires 1 \XOR{} query and $\bTnoparen(\sqrt{N})$ in-place queries.
    Then, we build on a recent extension of the quantum adversary method to
    characterize exact conditions for a decision problem to exhibit such a separation,
    and we propose a candidate problem.
\end{abstract}

\section{Introduction}\label{sec:intro}
Quantum algorithms are typically developed and characterized in terms of query complexity.
The strongest promises of quantum advantage over classical computation come from
unconditional separations proved in terms of black-box queries,
including Shor's period-finding algorithm and Grover's search algorithm.
Understanding the nuances of the query model is therefore essential for advancing quantum algorithm design and sculpting quantum advantages.

Given an arbitrary Boolean function $f$,
the standard query model in quantum computation is defined by \XOR{} oracles $\oracle{S}_f$,
also known as ``standard oracles'',
which map basis states $\ket{x}\ket{y}$ to $\ket{x}\ket{y\oplus f(x)}$.
Other common models, such as phase oracles, are known to be equivalent.
The use of \XOR{} oracles goes back to the early days of quantum computation
\cite{Feynman86-quantumComputers,Deutsch85,DeutschJozsa92,BernsteinVazirani-BV97-quantumComplexityTheory,BBBV97}
and even reversible computation
\cite{Bennett73-logicalReversibility,Bennett82-thermodynamicsOfComputation,Peres85-reversible,Bennett89-reversible}.
\XOR{} oracles embed potentially irreversible functions in a reversible way,
ensuring that all queries are unitary.
This enables quantum query complexity to encompass arbitrary Boolean functions
and offers a standard input-output format for using one algorithm as a sub-routine in another.

Other oracle models for quantum computation have been studied,
but most abandon unitarity
\cite{Regev-RS08-faultyOracle,Harrow-HR13-uselessOracle,Temme14-searchWithFaultyHamiltonianOracle,LL16-bombQueries,HG22-cqHybridSchemes,Rosmanis23-searchWithNoisyOracle,LMP24-faultyOracle,NatarajanNirkhe-NN24-QCMA}
or provide query access to quantum functions with no analogue in classical query complexity, \eg{} general unitaries \cite{Aaronson-AK07-QCMA,ABPS24}.

When querying a permutation, there is another natural oracle model: an \textit{in-place} oracle $\oracle{P}_f$ which maps $\ket{x}$ to $\ket{f(x)}$.
These oracles have been called
in-place \cite{Fefferman-FK18-QCMA,BFM23-nonstandardOracles},
erasing \cite{Aaronson02collisionProblem,Aaronson21-querySurvey},
and
minimal \cite{Kashefi-KKVB02comparisonQuantumOracles,Atici03-compareOracles}.\footnote{
    Unfortunately, ``permutation oracle'' has been used to refer to
    any oracle which embeds a permutation.

    Following a suggestion by John Kallaugher,
    we have found it convenient in conversation to refer to
    ``xoracles'' and ``smoracles'' (for ``small oracles'').
}
Just like \XOR{} oracles, in-place oracles can be directly studied and compared in both quantum and classical computation.

In-place oracles were first studied in the quantum setting by
Kashefi, Kent, Vedral, and Banaszek \cite{Kashefi-KKVB02comparisonQuantumOracles}.
They showed several results comparing \XOR{} oracles and in-place oracles,
including a proof that $\bTnoparen(\sqrt{N})$ queries to an \XOR{} oracle are required to simulate an in-place query to the same permutation.
Around the same time, Aaronson \cite{Aaronson02collisionProblem}
proved that \setComp{},
an approximate version of the \collision{} problem,
requires an exponential number of \XOR{} queries but only a constant number of in-place queries.

These oracles relate to multiple topics in  quantum algorithms and complexity theory.
Aaronson's lower bound for the collision problem \cite{Aaronson02collisionProblem} was partially inspired by the desire to separate the in-place and \XOR{} query models.
\cite{Kashefi-KKVB02comparisonQuantumOracles} observed that a constant number of in-place queries
is sufficient to solve \rigidGraphIso{}, a necessary subcase for solving general \graphIso{}.
An identical protocol was later generalized to define the concept of \prob{QSampling}, which is sufficient to solve \class{SZK},
by Aharonov and Ta-Shma \cite{Aharonov-AT03-adiabatic}.
These ideas inspired pursuing lower bounds on the \indexErasure{} problem \cite{Shi01-arxivVersion-collisionProblem,AMRR11-symmetryStateGeneration,Rosmanis-LR20-indexErasure},
ruling out potential algorithms for \graphIso{} using \XOR{} oracles.
Fefferman and Kimmel \cite{Fefferman-FK18-QCMA} showed an oracle separation of \QMA{} and \QCMA{} relative to randomized in-place oracles.
Also, the expressive power of in-place oracles relates to the conjectured existence of one-way permutations \cite[p. 926]{Bennett82-thermodynamicsOfComputation}.
Additionally, because in-place oracles are not self-inverse, they offer a setting in which to study computation with inverse-free gate sets \cite{Bouland-BG21-inverseFreeSolovayKitaev}.

In-place oracles outperform \XOR{} oracles in every established separation between the two query models,
but it is conjectured that the oracles are incomparable, each better-suited for certain tasks.
Aaronson \cite{Aaronson21-querySurvey} raised proving such a separation as an open problem.
Fefferman and Kimmel \cite{Fefferman-FK18-QCMA} conjectured that inverting a permutation over $N$ elements,
a task which requires only $\bOnoparen(\sqrt{N})$ queries to an \XOR{} oracle,
requires $\bOm{N}$ queries to an in-place oracle.
\permInvers{} is formally as hard as unstructured search \cite{Nayak11-permInversAsHardAsSearch},
so this conjecture effectively predicts that the speedup of Grover's algorithm \cite{Grover-G96-search}
is impossible with an in-place oracle.

\paragraph{Results}
We refute the conjecture of \cite{Fefferman-FK18-QCMA} by designing a new quantum algorithm that solves \permInvers{}
with $\bOnoparen(\sqrt{N})$ queries to an in-place oracle,
recovering the same speedup as Grover's search algorithm.

We additionally apply this algorithm to tightly characterize the ability of \XOR{}
and in-place oracles to simulate each other.
Then, we change focus and make progress towards showing the desired separation.
We introduced a subspace-conversion problem that requires 1 \XOR{} query and exponentially-many in-place queries.

Finally, we propose a candidate decision problem
that can be solved with $\bOnoparen(\sqrt{N})$ queries to an \XOR{} oracle
and that we conjecture requires $\bOm{N}$ queries to an in-place oracle.
We then apply recent advances in the quantum adversary bound to define a new class of adversary matrices which must be used if such a decision-problem separation exists.

\subsection{Quantum Search}\label{sec:intro-quantSearch}
Unstructured search, famously solved by Grover's algorithm with $\bOnoparen(\sqrt{N})$ queries to an \XOR{} oracle,
is one of the most well-studied problems in quantum query complexity.
The first non-trivial quantum lower bound was for unstructured search \cite{BBBV97}.
Later work modifying the query model, for instance by introducing noise or faults into queries, focused on unstructured search
\cite{Regev-RS08-faultyOracle,Temme14-searchWithFaultyHamiltonianOracle,LL16-bombQueries,HG22-cqHybridSchemes,Rosmanis23-searchWithNoisyOracle,ABPS24}.

In-place oracles are only defined for bijections (see \cref{sec:oracles}).
Restricted to permutations, the unstructured search problem is equivalent to \permInvers{}
\cite{Nayak11-permInversAsHardAsSearch}.\footnote{
    The reductions between \permInvers{} and unstructured search
    are entirely classical.
    So the reductions hold using either \XOR{} oracles or in-place oracles, although some quantum garbage registers may differ.
}
\begin{definition}\label{def:permInvers}
    Given query access to a permutation $\pi$ on $[N]=\{0,\dots,N-1\}$,
    the \permInvers{} problem is to output $\pi^{-1}(0)$.
\end{definition}

The choice to invert $0$ can of course be replaced with any element.
It is also straightforward to define a related decision problem,
for example, deciding if $\pi\inv(0)$ is odd or even.

Like general unstructured search,
\permInvers{} has been a frequent target for new lower bound techniques.
It can be solved with $\bOnoparen(\sqrt{N})$ queries to an \XOR{} oracle using Grover's algorithm.
Ambainis \cite{Ambainis02adversaryMethod} applied his new quantum adversary method to show that $\bOmnoparen(\sqrt{N})$
queries to an \XOR{} oracle are in fact required to solve the problem.
Nayak \cite{Nayak11-permInversAsHardAsSearch} gave an alternative proof by showing the problem is
as hard as general unstructured search.
Rosmanis \cite{Rosmanis22-tightboundsinvertingpermutations} also reproduced this tight lower bound using
the compressed oracle technique on random permutations.
As for in-place oracles,
\cite{Fefferman-FK18-QCMA} proved that
$\bOmnoparen(\sqrt{N})$ in-place queries are needed to solve \permInvers{}.
Belovs and Yolcu \cite{Belovs2023onewayticketlasvegas} later applied their advancements on the quantum adversary method to reprove the same lower bound.

We add to this sequence of work, studying \permInvers{} in \cref{sec:permInversion} to give the following result.

\begin{theorem}\label{thm:permInvers}
    For a permutation $\pi$ on $[N]$, \permInvers{} can be solved with $\bOnoparen(\sqrt{N})$ in-place queries to $\pi$.
\end{theorem}

Thus, we refute the conjecture that $\bOm{N}$ in-place queries are required,
and we show the
$\bOmnoparen(\sqrt{N})$ lower bound \cite{Fefferman-FK18-QCMA,Belovs15variationsOnQuantum} is tight.

\paragraph{Grover's Algorithm}
Before we sketch our algorithm, we first recall Grover's algorithm for unstructured search \cite{Grover-G96-search} in the context of \permInvers{}.
Grover's algorithm repeatedly alternates between using \XOR{} queries to negate the amplitude of $\ket{\pi^{-1}(0)}$
and using the ``Grover Diffusion operator'' to reflect all amplitudes about the average,
steadily amplifying $\ket{\pi^{-1}(0)}$ on every iteration.
In other words, the algorithm
alternates between the oracle-dependent reflection
$I - 2\ketbra{\pi^{-1}(0)}{\pi^{-1}(0)}$
and the diffusion reflection
\begin{equation}\label{eqn:diffusion}
    \diffusion{} = I - 2\ketbra{s}{s} ,
\end{equation}
where $\ket{s}$ is the uniform superposition $\frac{1}{\sqrt{N}}\sum\ket{i}$.
This is illustrated in \cref{fig:groverFigure}.

In-place oracles seem at odds with oracle-dependent reflections, since reflections---like \XOR{} queries---are self-inverse, but inverting an in-place query is equivalent to inverting the underlying permutation, which would solve \permInvers{}.
With this in mind, it would be natural to conjecture, as \cite{Fefferman-FK18-QCMA} did, that no Grover-style speedup is possible using in-place oracles.

\begin{figure}
    \centering
    \scalebox{1}{\input{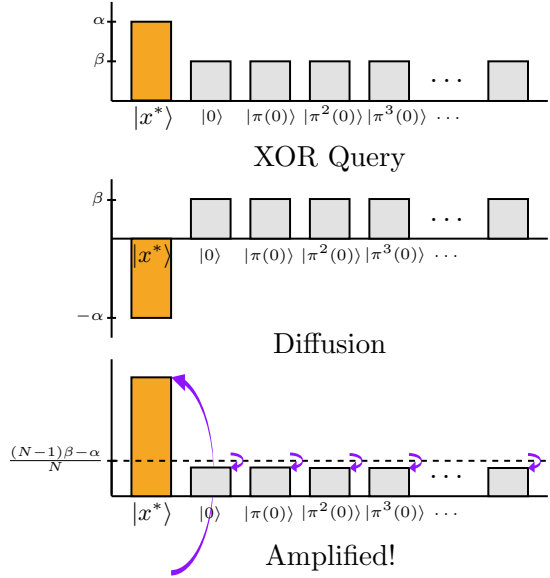}}
    \caption{(Color) Illustration of how one iteration of Grover's search algorithm amplifies $\ket{x^*} := \ket{\pi^{-1}(0)}$.}
    \label{fig:groverFigure}
\end{figure}

\paragraph{A New Algorithm}
Let $x^* \coloneqq \pi^{-1}(0)$ be the ``marked item" to be found.
Our algorithm starts with an equal superposition over $[N]$ along with an ancilla register and a ``flag'' qubit:
$\frac{1}{\sqrt{N}} \sum \ket{i}\ket{0^n}\ket{0}$.
The algorithm repeatedly iterates over steps \textit{Mark}, \textit{Shift}, and \textit{Diffuse the Difference}.
The intuition behind these steps is as follows.
\begin{itemize}
    \item \textit{Mark:}
    Query the oracle on all $i\in [N]$.
    Then,
    conditioned on the output of $\pi(i)$ being $0$,
    flip the flag qubit from $\ket{0}$ to $\ket{1}$.

    (The \textit{Mark} step cannot be used to implement Grover's algorithm as usual
    because the query answer remains in the ancilla register, as garbage, until the next step.)

    \item \textit{Shift:} In the $\ket{1}$-flagged branch, all amplitude is concentrated on $\ket{x^*}$, while in the $\ket{0}$-flagged branch, the amplitude is spread evenly over all basis states except $\ket{x^*}$.

    In only the $\ket{0}$-flagged branch of the superposition,
    query the oracle to shift the amplitude of each basis state forward according to $\pi$
    (perform a controlled in-place query to $\pi$). \\
    This shifts amplitude from $\ket{i}$ onto $\ket{\pi(i)}$, and in particular, from $\ket{\pi^{-1}(x^*)}$ onto $\ket{x^*}$.

    \item \textit{Diffuse the Difference:}
    The two branches are now such that if they are interfered
    to produce two branches, one branch which adds amplitudes and another branch which subtracts amplitudes,
    then the amplitude on $\ket{x^*}$ would be above average in the former branch and below average in the latter branch.

    Perform the standard Grover diffusion operator (\cref{eqn:diffusion})
    controlled on the flag qubit being the $\ket{-}$ state,
    which reflects the ``difference branch'' about its average amplitude.

    This results in the amplitude on $\ket{x^*}$ being similarly amplified in both branches.
    In fact, we find the branches are inverse-exponentially close to each other,
    and that after the $t$-th iteration, the overall state is effectively
    \begin{equation*}
        \ket{\psi_{t}} =  \left( \alpha_{t} \ket{x^*} + \sum_{i \in [N] \setminus \set{x^*}} \beta_{t} \ket{i} \right) \ket{0^n}\ket{0} ,
    \end{equation*}
    where $\alpha_t$ increases by approximately $1/\sqrt{N}$ each iteration.
\end{itemize}
These steps are repeated $\bOnoparen(\sqrt{N})$ times to amplify the amplitude on $\ket{x^*}$ until there is a constant probability of measuring it.
Each iteration uses a constant number of in-place queries,
so the overall query complexity is $\bOnoparen(\sqrt{N})$.
For more intuition, see a circuit diagram in \cref{fig:PermInv qckt} on \cpageref{fig:PermInv qckt}
and an illustration in \cref{fig:algFigure} on \cpageref{fig:algFigure} similar to \cref{fig:groverFigure} above.

In \cref{sec:controlledinplace}, we give a construction for the controlled in-place query necessary for the \textit{Shift} step of the algorithm.
This construction differs significantly from the analogous construction for \XOR{} oracles.

\begin{lemma}\label{result:contPi}
    There exists a unitary circuit making 1 in-place query to $\pi$
    which for all $x\in [N]$ maps
    \[
    \ket{a} \ket{x} \ket{y} \mapsto
    \begin{cases}
        \ket{a}\ket{x} \ket{y} & \text{when } a = 0 \\
        \ket{a}\ket{\pi(x)} \ket{y} & \text{when } a = 1
    \end{cases} ,
    \]
    where $y$ is the image under $\pi$ of some fixed point, such as $y=\pi(0)$.
\end{lemma}

Note that although $y$ depends on the oracle $\pi$,
it is independent of the query $x$.
So while $y$ is garbage, it is effectively negligible.
Because it is never entangled with the input register,
the garbage can be safely measured and erased.
See \cref{sec:controlledinplace} for more details.

\subsection{Simulating Other Oracles}\label{sec:intro-otherOracles}
In \cref{sec:otherOracles}, we tightly characterize the ability of \XOR{} and in-place oracles to simulate each other.
We do so by
applying our new algorithm to give new upper bounds and by developing a novel lower bound.

For a permutation $\pi$ on $[N]$,
Grover's algorithm can be used to simulate an \XOR{} query to $\pi^{-1}$, an in-place query to $\pi^{-1}$,
or an in-place query to $\pi$ using $\bOnoparen(\sqrt{N})$ \XOR{} queries to $\pi$,
and this complexity is known to be tight \cite{Kashefi-KKVB02comparisonQuantumOracles}.
We show how to use our new algorithm to perform the analogous simulations
using $\bOnoparen(\sqrt{N})$ queries to an in-place oracle.
The constructions are non-trivial due to the inability of in-place oracles to uncompute garbage.
The simulations are approximate with inverse-exponential error due to the error in our algorithm for \permInvers{}.

Next, we prove that our simulations are tight by giving matching lower bounds.
Inspired by \cite{Kashefi-KKVB02comparisonQuantumOracles},
we prove this by arguing that if few in-place queries could simulate an \XOR{} query, then we could violate the lower bound of \cite{Fefferman-FK18-QCMA} for performing unstructured search.

\begin{theorem}\label{thm:inPlaceSimulatingXOR_lower}
    For a permutation $\pi$ on $[N]$,
    $\bOmnoparen(\sqrt{N})$ in-place queries to $\pi$ are necessary to approximately simulate an \XOR{} query to $\pi$.
\end{theorem}

Given that an \XOR{} query to $\pi$ can be implemented using 1 \XOR{} query to $\pi$,
\cref{thm:inPlaceSimulatingXOR_lower} makes this the first task known to require more in-place queries than \XOR{} queries.
We improve on this in the next section.

We can summarize all upper and lower bounds above as follows.

\begin{restatable}[Summary of relationships]{corollary}{summary}
\label{cor:summaryRelationships}
    For a permutation $\pi$ on $[N]$,
    $\bTnoparen(\sqrt{N})$ queries to any one of
    an in-place oracle for $\pi$, an in-place oracle for $\pi\inv$, an \XOR{} oracle for $\pi$, or an $\XOR{}$ oracle for $\pi\inv$
    are necessary and sufficient to approximately simulate
    any one of the others.
\end{restatable}

\subsection{A Subspace-Conversion Separation}\label{sec:intro-subspace}
Next, in \cref{sec:subspace} we improve the unitary-implementation separation given in the previous section to a subspace-conversion separation.

\indexErasure{} is the task of generating the state $\frac{1}{\sqrt{N}}\sum_{x\in [N]}\ket{f(x)}$ given queries to $f$.
It was introduced by Shi \cite{Shi01-arxivVersion-collisionProblem}
and formalized as a state-generation task by
Ambainis, Magnin, Roetteler, and Roland \cite{AMRR11-symmetryStateGeneration}.
As noted by \cite{Shi01-arxivVersion-collisionProblem},
solving \indexErasure{} would imply solutions to \setEquality{} and \graphIso{}. Similar work on QSampling \cite{Aharonov-AT03-adiabatic} suggests many more applications.
\indexErasure{} requires $\bOmnoparen(\sqrt{N})$ \XOR{} queries \cite{AMRR11-symmetryStateGeneration,Rosmanis-LR20-indexErasure} but just 1 in-place query, so
the problem seems to capture key differences between the models.

We define the converse problem, \functionErasure{}.

\begin{definition}\label{def:funcErasure}
    Given query access to a function $f$, \functionErasure{} is the subspace-conversion problem of
    transforming any superposition
    of the form $\sum \alpha_x \ket{x}\ket{f(x)}$ to
    $\sum \alpha_x \ket{x}$.
\end{definition}

A state-conversion problem requires implementing an algorithm
which, given an oracle to function $f$,
maps an input $\ket{\psi_f}$ to output $\ket{\phi_f}$.
A subspace-conversion problem simply
generalizes this to multiple input-output pairs for each oracle function $f$.
We discuss the details of unitary-implementation, subspace-conversion, and other types of problems in \cref{sec:subspace}.

\functionErasure{} can trivially be solved with 1 \XOR{} query to $f$.
Then by \cref{cor:summaryRelationships},
$\bOnoparen(\sqrt{N})$ in-place queries are sufficient.
Finally, we show how \functionErasure{} and one additional in-place query are sufficient to simulate an \XOR{} query.
To avoid violating
\cref{thm:inPlaceSimulatingXOR_lower}, this implies $\bOmnoparen(\sqrt{N})$ queries are necessary.

\begin{theorem}\label{thm:inPlaceAndFE}
    For a permutation $\pi$ on $[N]$,
    $\bTnoparen(\sqrt{N})$ in-place queries to $\pi$ are necessary and sufficient for \functionErasure{}.
\end{theorem}

\cref{thm:inPlaceAndFE} makes \functionErasure{} the first coherent subspace-conversion problem known to require
fewer \XOR{} queries than in-place queries.
This improves on the new unitary-implementation separation from the previous section.

\subsection{Lower Bounds}\label{sec:intro-lower}
The first works to study in-place oracles proved that there are problems
which can be solved with asymptotically fewer queries to in-place oracles than to the corresponding \XOR{} oracles
\cite{Kashefi-KKVB02comparisonQuantumOracles,Aaronson02collisionProblem}.
They left open the question of whether a separation could be shown in the opposite direction,
making the two oracles formally incomparable,
or whether in-place oracles are generically superior to \XOR{} oracles.
Our main result (\cref{thm:permInvers}) refutes one conjectured path towards constructing a problem for which \XOR{} oracles are better than in-place oracles.
Our study of \functionErasure{}
demonstrates the first problem which provably requires fewer queries to an \XOR{} oracle than an in-place oracle,
although it is a subspace-conversion problem instead of a decision problem.
In \cref{sec:lowerbounds}, we consider the possibility of improving this to a decision-problem separation.

\paragraph{Conventional Lower Bound Techniques}
In \cref{sec:lowerbounds-conventionalTechniques}, we discuss how common quantum lower bound techniques, the polynomial method
\cite{BBCMdW01-polynomialMethod}
and the unweighted adversary method \cite{Ambainis02adversaryMethod}, fail to prove the desired separation.
We show that under these techniques, any lower bound on the number of in-place queries implies the same lower bound on the number of \XOR{} queries,
making these techniques unable to prove a separation where \XOR{} oracles outperform in-place oracles.

\paragraph{A Candidate Decision Problem}
In \cref{sec:candidateProblem}, we introduce a new problem, \perminvgarb{},
which can be solved with $\bTnoparen(\sqrt{N})$ queries to an \XOR{} oracle
and which we conjecture requires $\bOmnoparen(N)$ queries to an in-place oracle.
As we discuss, the problem is designed to embed an injection from $[N^2]$ to $[N]$
into a bijection on $[N^2]$, which we believe circumvents algorithms using in-place oracles.
The idea behind this problem builds on the ``Simon's problem with garbage'' proposed by Aaronson \cite{Aaronson21-querySurvey}.

\paragraph{Techniques for a Decision-Problem Separation}
Finally, in \cref{sec:techAndCondForSeparation},
we briefly discuss the potential for more sophisticated lower bound methods
to prove a decision-problem separation,
including for our candidate \perminvgarb{}.
A full exposition is given in \cref{app:adv}.

The recent extension of the quantum adversary method by
Belovs and Yolcu \cite{Belovs2023onewayticketlasvegas}
applies to arbitrary linear transformations, including in-place oracles.
The adversary bound is an optimization problem
over \textit{adversary matrices}
such that the optimal value equals the quantum query complexity for a given problem.
Of course, the difficulty with the adversary method is to design a ``good'' adversary matrix
exhibiting a tight bound.

We introduce a special class of feasible solutions which we call \textit{extended adversary matrices}.
We show, with some technical caveats, that there exists an \XOR{} query advantage over in-place oracles
for a decision problem if and only if it is witnessed by extended adversary matrices.
Then, for our candidate problem \perminvgarb{},
we are able to remove these caveats and state that if our conjectured separation is true,
then it must be witnessed by extended adversary matrices.

\subsection{Open Problems}\label{intro-open}
Our work suggests several topics
for improving on our results or understanding in-place oracles.
Several of these questions can be asked regarding
either quantum or \emph{classical reversible} computing.
\begin{enumerate}
    \item Variants of Grover's algorithm are known for multiple targets.
        Applications of Grover's algorithm include
        amplitude amplification, approximate counting, more.
        Can these variations or applications be reproduced with in-place oracles?
    \item Can other primitives of quantum computation, such as phase estimation, be reproduced with in-place oracles?
    \item What is the query complexity of
        $\perminvgarb$?
        We conjecture that a linear number of queries to an in-place oracle is necessary, separating them from \XOR{} oracles.
        Can another decision problem give a separation?
        Aaronson \cite{Aaronson21-querySurvey} has suggested a variant of Simon's problem
        that he suggests may relate to the effects of decoherence.
    \item Is a version of in-place oracles querying injections and implemented by isometries,
        rather than bijections and unitaries, interesting?
        The adversary method of \cite{Belovs2023onewayticketlasvegas}
        seems to still apply in this case.
    \item Bennett \cite{Bennett89-reversible} found there are provable time/space tradeoffs for reversible computation.
        Perhaps similar separations can be proved
        for in-place versus \XOR{} oracles.
\end{enumerate}
\section{Quantum Oracles}\label{sec:oracles}
As stated previously, the standard query model in quantum computation and
classical reversible computation is the \XOR{} oracle. Other common models, such as the phase oracle, are equivalent.
For a function $f$, an \XOR{} oracle $\oracle{S}_f$ maps
$\ket{x}\ket{y}\mapsto \ket{x}\ket{y\oplus f(x)}$,
where $\oplus$ denotes bitwise \XOR{} with queries encoded in binary.

When querying an invertible function, there is another natural unitary query model.\footnote{
    We restrict our study to bijections, and without loss of generality to permutations on $[N]$.
    A similar oracle which queries an injection would still be reversible,
    but it would be an isometry rather than a unitary.
    Our algorithm seems to require a bijection since is uses the oracle's previous outputs as its next inputs.
    \label{footnote:injections}
}
An in-place oracle $\Ppi$ maps $\ket{x} \mapsto \ket{\pi(x)}$.

Here we list several basic identities given by \cite{Kashefi-KKVB02comparisonQuantumOracles}.

\begin{enumerate}
    \item Given query access to both $\pi$ and $\piinv$, standard and in-place oracles are equivalent.\\
    More precisely, $\Ppi$ can be simulated using 1 query to $\Spi$ and 1 query to either of $\Spiinv,\Ppiinv$.
    Similarly, $\Spi$ can be simulated using 1 query to $\Ppi$ and 1 query to either of $\Spiinv,\Ppiinv$.
    So, the interesting case is when we can query $\pi$ but cannot query its inverse.

    \item \XOR{} oracles are self-inverse, $\Spi = \paren{\Spi}^\dagger$, but generally $\paren{\Spi}^\dagger\neq\Spiinv$.\\
    In contrast, generally $\Ppi\neq\paren{\Ppi}^\dagger$ but it does hold that $\paren{\Ppi}^\dagger = \Ppiinv$.

    \item $\bTnoparen(\sqrt{N})$ queries to an \XOR{} oracle $\Spi$ can be used to simulate a query to $\Ppi$.\\
    The upper bound is due to Grover's search algorithm.
    The lower bound follows by
    observing that a circuit for $\Ppi$ querying $\Spi$ can be inverted to give a circuit for $\Ppiinv$
    querying $\paren{\Spi}^\dagger = \Spi$,
    which would solve \permInvers{}, which requires $\bOmnoparen(\sqrt{N})$ queries to $\Spi$.
\end{enumerate}

The \XOR{} query model was motivated by two needs.
First is the need to embed non-invertible functions in
a reversible query.
Second is that because \XOR{} oracles are self-inverse, they enable uncomputing.
An early criticism of reversible computation by Landauer \cite{Landauer61} was that in order to maintain reversibility,
a computation would need to retain intermediate work until the end,
only deferring the cost of information erasure instead of avoiding it.
To the contrary, Bennett showed that any circuit can efficiently be made into a reversible one that uncomputes any intermediate work and gives its original output in the form of an \XOR{} query \cite{Bennett73-logicalReversibility,Bennett82-thermodynamicsOfComputation}.
Given a garbage-producing reversible circuit, first apply the circuit,
then copy the desired output into a new register using \XOR{},
and then apply the circuit in reverse, gate-by-gate, to uncompute all intermediate steps,
leaving only the input and the copied output.
Moreover, such a gate-by-gate reversal works when one algorithm is used as a black-box subroutine for another,
since given a black-box following this \XOR{}-model,
it is self-inverse.
So full algorithms, including subroutines, can indeed be reversed gate-by-gate.
Besides these two reasons,
\XOR{} oracles simply appeared natural at the time quantum computing was formalized.
As far as we are aware, in-place oracles have not been studied in the classical reversible computing literature.
There have been just a few references to alternative classical reversible implementations of 1-to-1 functions
\cite{Peres85-reversible,Bennett89-reversible}.
So quantum computation, which is based on reversible operations,
later inherited the \XOR{} model.
At the same time,
the ability to uncompute enabled quantum interference
\cite{Feynman86-quantumComputers,BernsteinVazirani-BV97-quantumComplexityTheory}.
Many early results also only involved binary functions,
and other results were motivated more by ensuring quantum computers could implement
tasks such as error-reduction and subroutines ($\BQP^\BQP=\BQP$ \cite{BBBV97})
rather than questioning the query model.

One more important feature of \XOR{} oracles is that for a function $f$,
the complexity of implementing $\oracle{S}_f$
using reversible operations
is at most a constant multiplicative factor more than the
general, irreversible circuit implementing $f$ \cite{Bennett73-logicalReversibility}.
For in-place oracles,
no construction is known for efficiently transforming
an irreversible circuit for permutation $\pi$ into a reversible circuit for $\Ppi$.
In fact, the widely believed existence of one-way permutations implies that there exist
permutations for which this is impossible.
This is because given a reversible implementation of $\Ppi$,
inverting the circuit gate-by-gate gives $\paren{\Ppi}^\dagger = \Ppiinv$ with exactly the same circuit size,
whereas one-way permutations should have different complexities than their inverses.
This may limit the practical instantiation of in-place oracles,
although they may lead to useful insights in other ways.

\subsection{Controlled In-Place Oracle} \label{sec:controlledinplace}
Given an \XOR{} oracle $\oracle{S}_f$, it is easy to implement the controlled oracle $\cont\oracle{S}_f$:
query $\oracle{S}_f$, conditionally save the result, and then uncompute by querying $\oracle{S}_f$ again.
For general unitary oracles, access to controlled queries is non-trivial.
In fact, it is impossible to implement the controlled version of an unknown unitary given only black-box access \cite{AFCB14-controlledOperations,TMVG18-controlledOperations}.
This may prevent implementing standard algorithms like phase estimation \cite{CCGMSS24-phaseEstimation}.

Here we give \cref{result:contPi}, that it is possible to efficiently implement a controlled in-place oracle,
although we produce some query-independent garbage.
We give a unitary algorithm which implements the map
\[
    \ket{a}_\regA \ket{x}_\regB \ket{\pi(0)}_\regC \mapsto
    \begin{cases}
        \ket{a}\ket{x} \ket{\pi(0)} & \text{when } a = 0 \\
        \ket{a}\ket{\pi(x)} \ket{\pi(0)} & \text{when } a = 1
    \end{cases} ,
\]
performing a controlled in-place query to $\pi$ so long as $\regC$ contains $\pi(0)$.
Here, 0 could be replaced with any fixed input.
The $\regC$ register,
which depends on the permutation $\pi$ but not on the query $x$,
can be prepared ahead of time with a single query to $\Ppi$.
In particular, because $\pi(0)$ is independent of the query $x$, an algorithm could save $\pi(0)$ to a single global register which is shared by all controlled queries to $\Ppi$.
This register is never entangled with the quantum state, so it can be safely erased at any time.
The register could even be prepared by classical preprocessing querying
$\Ppi$ and modifying the quantum circuit as necessary to automatically write and uncompute $\pi\paren{0}$ as needed.

\begin{proof}[Proof of \cref{result:contPi}]
Let $\regA$, $\regB$, and $\regC$ contain the control qubit $a$, the query $x$, and $\pi(0)$, respectively. Initialize an auxiliary register $\regD$ to $\ket{0}$. First, conditioned on $a$ being $1$, swap registers $\regB$ and $\regD$. Second, apply $\Ppi$ to $\regD$. Third, conditioned on $a$ being $0$, \XOR{} the contents of $\regC$ into $\regD$ to clear the auxiliary register. Finally, conditioned on $a$ being $1$, swap back registers $\regB$ and $\regD$.
Omitting register $\regC$, which is fixed to $\ket{\pi(0)}$ throughout, these steps are illustrated below.
\begin{align*}
    \text{Swap\hspace{3.5em}Query\hspace{3.8em}\XOR{}\hspace{4.5em}Swap\hspace{8.5em}}
    \\
    \ket{a}_\regA \ket{x}_\regB \ket{0}_\regD
    \rightarrow
    \begin{cases}
        \ket{a}\ket{x}\ket{0} \rightarrow \ket{a}\ket{x}\ket{\pi(0)} \rightarrow \ket{a}\ket{x}\ket{0}\hspace{1.4em}\rightarrow \ket{a}\ket{x}\ket{0}
        & \text{when } a = 0
        \\
        \ket{a}\ket{0}\ket{x} \rightarrow \ket{a}\ket{0}\ket{\pi(x)} \rightarrow \ket{a}\ket{0}\ket{\pi(x)} \rightarrow \ket{a}\ket{\pi(x)}\ket{0}
        & \text{when } a = 1
    \end{cases}
\end{align*}
So, rather than query and conditionally copy as we can for the \XOR{} oracle, here we conditionally query either $x$ or some fixed value.
\end{proof}
\section{Permutation Inversion}\label{sec:permInversion}
In this section, we prove our main result, that \permInvers{} (\cref{def:permInvers}) can be solved with $\bTnoparen(\sqrt{N})$ queries to an in-place oracle.

\begin{proof}[Proof of \cref{thm:permInvers}]
    The lower bound was proved by Fefferman and Kimmel \cite{Fefferman-FK18-QCMA} and later reproved by \cite{Belovs2023onewayticketlasvegas}.
    To prove the upper bound, we give an algorithm.

    \paragraph{Algorithm}
    For convenience, we assume $N = 2^n$ and identify the integers $[N]$ by their binary representations in $\{0,1\}^n$. We denote the target element $\pi^{-1}(0)$ by $x^*$.

    First, query $\Ppi$ once to check whether $\pi(0)$ is 0, and terminate early with answer 0 if it is.
    Otherwise,
    initialize three registers to the state
    $\ket{\psi_0} := \frac{1}{\sqrt{N}}\sum_{i=1}^N \ket{i}_{\regA} \ket{0^n}_{\regB} \ket{0}_{\regC}$,
    where $\regA$ and $\regB$ are each $n=\log{N}$ qubits and $\regC$ is one qubit.
    Then, repeat the following steps $T=\bOnoparen(\sqrt{N})$ times:
    \begin{enumerate}
        \item[1.] \textit{Mark} \\
        \XOR{} register $\regA$ into $\regB$, and apply $\Ppi$ to $\regB$. \\
        Controlled on $\regB$ being $\ket{0^n}$, apply $\NOT$ to $\regC$, flagging the branch where $\regA$ contains $x^*$.

        \item[2.] \textit{Shift (and Clean Up)} \\
        Controlled on $\regC$ being $\ket{0}$, apply $\Ppi$ to $\regA$. \\
        Controlled on $\regC$ being $\ket{0}$, \XOR{} $\regA$ into $\regB$, resetting $\regB$ to $\ket{0^n}$.

        \item[3.] \textit{Diffuse the Difference} \\
        Controlled on $\regC$ being $\ket{-}$, apply the diffusion operator to $\regA$.\\
        The diffusion operator $\diffusion{} \coloneqq 2 H^{\otimes n} \ketbra{0^n}{0^n} H^{\otimes n} - I$ is the same used in Grover's algorithm \cite{Grover-G96-search},
        equivalent to a reflection about the uniform superposition.

        \item[4.] Optional: \textit{Measure} \\
        Measure $\regC$. If $\ket{1}$ is observed then abort and report failure.
    \end{enumerate}
    Finally, measure register $\regA$ and output the result.
    See \cref{fig:PermInv qckt} for a circuit diagram
    of one iteration of the algorithm
    and \cref{fig:algFigure} for an illustration of the effect.

    Below, we will find that each \textit{Measure} step aborts with probability $1/N$.
    So, these intermediate measurements could be omitted and the qubit reused as it is,
    and the quantum union bound \cite{Gao15unionBound,ODonnell-OV22quantumunionbound}
    implies the overall success probability would decrease by at most $\sqrt{T/N} = \bO{N^{-1/4}}$. For now, we include the optional \textit{Measure} step to simplify the analysis.

    \begin{figure}[h]
    \centering
    \[
    \Qcircuit @C=.5em @R=0.5em @!R {
         & & & & & & \mbox{\hspace{1.8em}\textit{Mark}} & & & & & & & & \mbox{\textit{Shift}} & & & & & & & & \mbox{\textit{Diffuse the Difference}} \\
    	& \lstick{} & \qw       & \push{\rule{0em}{1em}} \qw & \ctrl{4}  & \qw       & \qw       & \qw       & \qw       & \push{\rule{0em}{1em}} \qw & \push{\rule{0em}{1em}} \qw & \push{\rule{0em}{1em}} \qw & \push{\rule{0em}{1em}} \qw & \multigate{2}{\Ppi} & \ctrl{3}  & \qw       & \qw       & \push{\rule{0em}{1em}} \qw & \push{\rule{0em}{1em}} \qw & \push{\rule{0em}{1em}} \qw & \push{\rule{0em}{1em}} \qw & \qw       & \multigate{2}{\diffusion} & \qw       & \push{\rule{0em}{1em}} \qw & \push{\rule{0em}{1em}} \qw & \qw       & \qw       & \push{\rule{0em}{1em}}\\
    	& \lstick{} & \qw       & \push{\rule{0em}{1em}} \qw & \qw       & \ctrl{4}  & \qw       & \qw       & \qw       & \push{\rule{0em}{1em}} \qw & \push{\rule{0em}{1em}} \qw & \push{\rule{0em}{1em}} \qw & \push{\rule{0em}{1em}} \qw & \ghost{\Ppi} & \qw       & \ctrl{2}  & \qw       & \push{\rule{0em}{1em}} \qw & \push{\rule{0em}{1em}} \qw & \push{\rule{0em}{1em}} \qw & \push{\rule{0em}{1em}} \qw & \qw       & \ghost{\diffusion} & \qw       & \push{\rule{0em}{1em}} \qw & \push{\rule{0em}{1em}} \qw & \qw       & \qw       & \rstick{\ket{\psi_{t}}_\regA} \push{\rule{0em}{1em}}\\
    	& \lstick{} & \qw       & \push{\rule{0em}{1em}} \qw & \qw       & \qw       & \ctrl{4}  & \qw       & \qw       & \push{\rule{0em}{1em}} \qw & \push{\rule{0em}{1em}} \qw & \push{\rule{0em}{1em}} \qw & \push{\rule{0em}{1em}} \qw & \ghost{\Ppi} & \qw       & \qw       & \ctrl{1}  & \push{\rule{0em}{1em}} \qw & \push{\rule{0em}{1em}} \qw & \push{\rule{0em}{1em}} \qw & \push{\rule{0em}{1em}} \qw & \qw       & \ghost{\diffusion} & \qw       & \push{\rule{0em}{1em}} \qw & \push{\rule{0em}{1em}} \qw & \qw       & \qw       & \push{\rule{0em}{1em}}
    \inputgroupv{2}{4}{0.5em}{2em}{\ket{\psi_{t-1}}_\regA \hspace{1.3em}} \\
    	& \lstick{\ket{0}_\regC \Big\{} & \qw       & \push{\rule{0em}{1em}} \qw & \qw       & \qw       & \qw       & \qw       & \targ     & \push{\rule{0em}{1em}} \qw & \push{\rule{0em}{1em}} \qw & \push{\rule{0em}{1em}} \qw & \push{\rule{0em}{1em}} \qw & \ctrlo{-1} & \ctrlo{1} & \ctrlo{2} & \ctrlo{3} & \push{\rule{0em}{1em}} \qw & \push{\rule{0em}{1em}} \qw & \push{\rule{0em}{1em}} \qw & \push{\rule{0em}{1em}} \qw & \gate{H}  & \ctrl{-1} & \gate{H}  & \push{\rule{0em}{1em}} \qw & \push{\rule{0em}{1em}} \qw & \meter    & \cw       & \rstick{\text{\textit{Measure} }\regC} \\
    	& \lstick{} & \qw       & \push{\rule{0em}{1em}} \qw & \targ     & \qw       & \qw       & \multigate{2}{\Ppi} & \ctrlo{-1} & \push{\rule{0em}{1em}} \qw & \push{\rule{0em}{1em}} \qw & \push{\rule{0em}{1em}} \qw & \push{\rule{0em}{1em}} \qw & \qw       & \targ     & \qw       & \qw       & \push{\rule{0em}{1em}} \qw & \push{\rule{0em}{1em}} \qw & \push{\rule{0em}{1em}} \qw & \push{\rule{0em}{1em}}\\
    	& \lstick{} & \qw       & \push{\rule{0em}{1em}} \qw & \qw       & \targ     & \qw       & \ghost{\Ppi} & \ctrlo{-1} & \push{\rule{0em}{1em}} \qw & \push{\rule{0em}{1em}} \qw & \push{\rule{0em}{1em}} \qw & \push{\rule{0em}{1em}} \qw & \qw       & \qw       & \targ     & \qw       & \push{\rule{0em}{1em}} \qw & \push{\rule{0em}{1em}} \qw & \push{\rule{0em}{1em}} \qw & \rstick{\ket{0^n}_\regB} \push{\rule{0em}{1em}}\\
    	& \lstick{} & \qw       & \push{\rule{0em}{1em}} \qw & \qw       & \qw       & \targ     & \ghost{\Ppi} & \ctrlo{-1} & \push{\rule{0em}{1em}} \qw & \push{\rule{0em}{1em}} \qw & \push{\rule{0em}{1em}} \qw & \push{\rule{0em}{1em}} \qw & \qw       & \qw       & \qw       & \targ     & \push{\rule{0em}{1em}} \qw & \push{\rule{0em}{1em}} \qw & \push{\rule{0em}{1em}} \qw & \push{\rule{0em}{1em}}
    \inputgroupv{6}{8}{0.5em}{2em}{\ket{0^n}_\regB \hspace{0.5em}}
    \gategroup{2}{29}{4}{29}{0em}{\}}
    \gategroup{6}{21}{8}{21}{0em}{\}}
    \gategroup{2}{4}{8}{10}{0.7em}{--}
    \gategroup{2}{13}{8}{18}{0.7em}{--}
    \gategroup{2}{21}{5}{25}{0.7em}{--}\\
    }
    \]
    \caption{
        One iteration of our \permInvers{} algorithm.
        $\diffusion$ is the standard diffusion operator. $\bullet$ denotes an operation controlled on $\ket{1}$ and $\circ$ denotes an operation controlled on $\ket{0}$.
    }
    \label{fig:PermInv qckt}
    \end{figure}
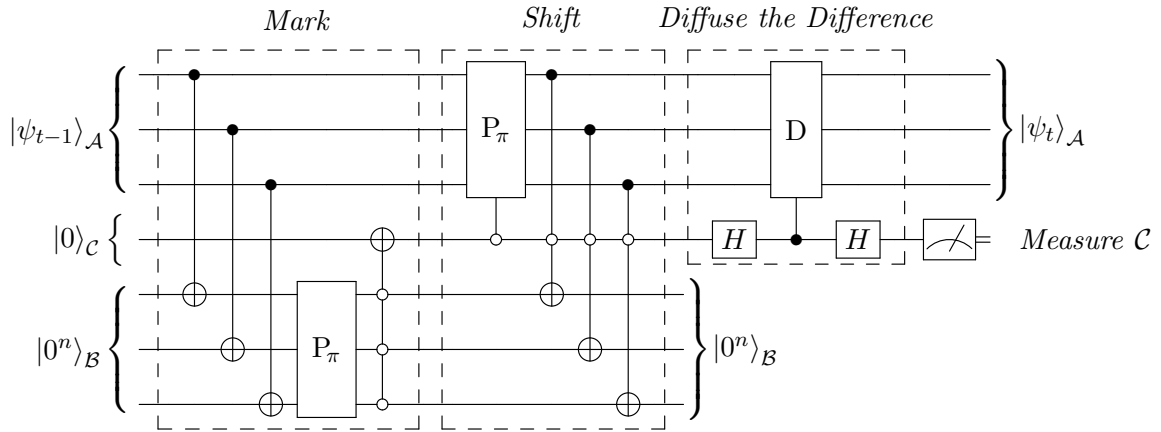

    \paragraph{Analysis}
    Now we prove that our algorithm succeeds with high probability.

    We use $\ket{\psi_t}$ to denote the state after $t$ iterations.
    We will show by induction that after each iteration, if the algorithm did not terminate early,
    then the state is of the form
    \begin{equation}\label{eq:psitForm}
        \ket{\psi_t} =  \left(\alpha_{t} \ket{x^*} + \sum_{i \in [N] \setminus \set{x^*}} \beta_{t} \ket{i}\right)_{\regA} \otimes \ket{0^n}_{\regB} \ket{0}_{\regC}
    \end{equation}
    for some real values $\alpha_t,\beta_t$. In particular, all $\ket{i}$ for $i\neq x^*$ share the same amplitude.
    The transformation from $\ket{\psi_{t-1}}$ to $\ket{\psi_t}$ is illustrated in \cref{fig:algFigure}.

    By construction, the initial state $\ket{\psi_0}$ is the uniform superposition,
    with $\alpha_0=\beta_0=\frac{1}{\sqrt{N}}$.

    Next, the $t$-th iteration begins with the state
    \[
        \ket{\psi_{t-1}} =  \left( \alpha_{t-1} \ket{x^*} + \sum_{i \in [N] \setminus \set{x^*}} \beta_{t-1} \ket{i} \right) \ket{0^n}\ket{0} .
    \]
    For ease of notation, we will drop the subscripts so that $\alpha,\beta$ implicitly refer to $\alpha_{t-1},\beta_{t-1}$.
    After the \textit{Mark} step, the state will be
    \[
        \ket{\psi'_{t-1}} = \alpha \ket{x^*} \ket{0^n} \ket{1}  + \sum_{i \in [N] \setminus \set{x^*}} \beta \ket{i} \ket{\pi(i)} \ket{0} .
    \]

    After the \textit{Shift (and Clean Up)} step, the state will be
    \begin{alignat*}{2}
        \ket{\psi''_{t-1}}
        & = \alpha \ket{x^*} \ket{0^n} \ket{1} + \sum_{i \in [N] \setminus \set{x^*}} && \beta \ket{\pi(i)} \ket{0^n} \ket{0}
        \\
        & = \alpha \ket{x^*} \ket{0^n} \ket{1} + \sum_{i \in [N] \setminus \set{0}} && \beta \ket{i} \ket{0^n} \ket{0} .
    \end{alignat*}
    As the name suggests, this step shifts amplitudes within the summation off of $\ket{0}$ and onto $\ket{x^*}$.

    Next, to prepare for the \textit{Diffuse the Difference} step, we rewrite register $\regC$ in the Hadamard basis.
    The state is equivalent to
    \begin{align*}
        \ket{\psi''_{t-1}} & = \phantom{+}
        \frac{1}{\sqrt{2}} \left[ \left(\beta+\alpha\right) \ket{x^*} + \sum_{i \in [N] \setminus \set{0,x^*}} \beta \ket{i} \right] \ket{0^n} \ket{+}
        \\
        & \phantom{=} +
        \frac{1}{\sqrt{2}} \left[ \left(\beta-\alpha\right) \ket{x^*} + \sum_{i \in [N] \setminus \set{0,x^*}} \beta \ket{i} \right] \ket{0^n} \ket{-} .
    \end{align*}
    Next, the \textit{Diffuse the Difference} step applies the diffusion operator \diffusion{} controlled on $\regC$ being $\ket{-}$.
    The diffusion operator can be viewed as reflecting every amplitude about the average amplitude.
    This results in
    \begin{gather*}
        \ket{\psi'''_{t-1}} =
        \frac{1}{\sqrt{2}} \Biggl[\left(\beta+\alpha\right) \ket{x^*} + \sum_{i \in [N] \setminus \set{0,x^*}} \beta \ket{i} \Biggr] \ket{0^n} \ket{+}
        \\
        + \frac{1}{\sqrt{2}} \Biggl[
            \left(\beta + \alpha - \frac{2 (\beta+\alpha) }{N}\right) \ket{x^*}
            + \left(2 \beta - \frac{2 (\beta+\alpha) }{N}\right) \ket{0}
            + \sum_{i \in [N] \setminus \{0,x^*\}} \left(\beta - \frac{2 (\beta+\alpha) }{N}\right) \ket{i}
        \Biggr]\!\ket{0^n}\!\ket{-} .
    \end{gather*}
    Returning register $\regC$ to the standard basis, we see
    \begin{align*}
        \ket{\psi'''_{t-1}} & =
        \left[ \left(\beta + \alpha - \frac{\beta+\alpha}{N}\right)\ket{x^*}
        + \sum_{i \in [N] \setminus \set{x^*}} \left(\beta - \frac{\beta+\alpha}{N}\right) \ket{i} \right] \ket{0^n}\ket{0}
        \\
        & + \left[
            \frac{\beta+\alpha}{N}\ket{x^*}
            - \left(\beta - \tfrac{\beta+\alpha}{N}\right)\ket{0}
            + \sum_{i \in [N] \setminus \set{0,x^*}} \frac{\beta+\alpha}{N} \ket{i}
        \right] \ket{0^n}\ket{1} .
    \end{align*}
   	The amplitude on $\ket{x^*}$ is now larger than the original amplitude $\alpha$.

    \begin{figure}[!b]
    	\centering
    	\scalebox{0.88}{\input{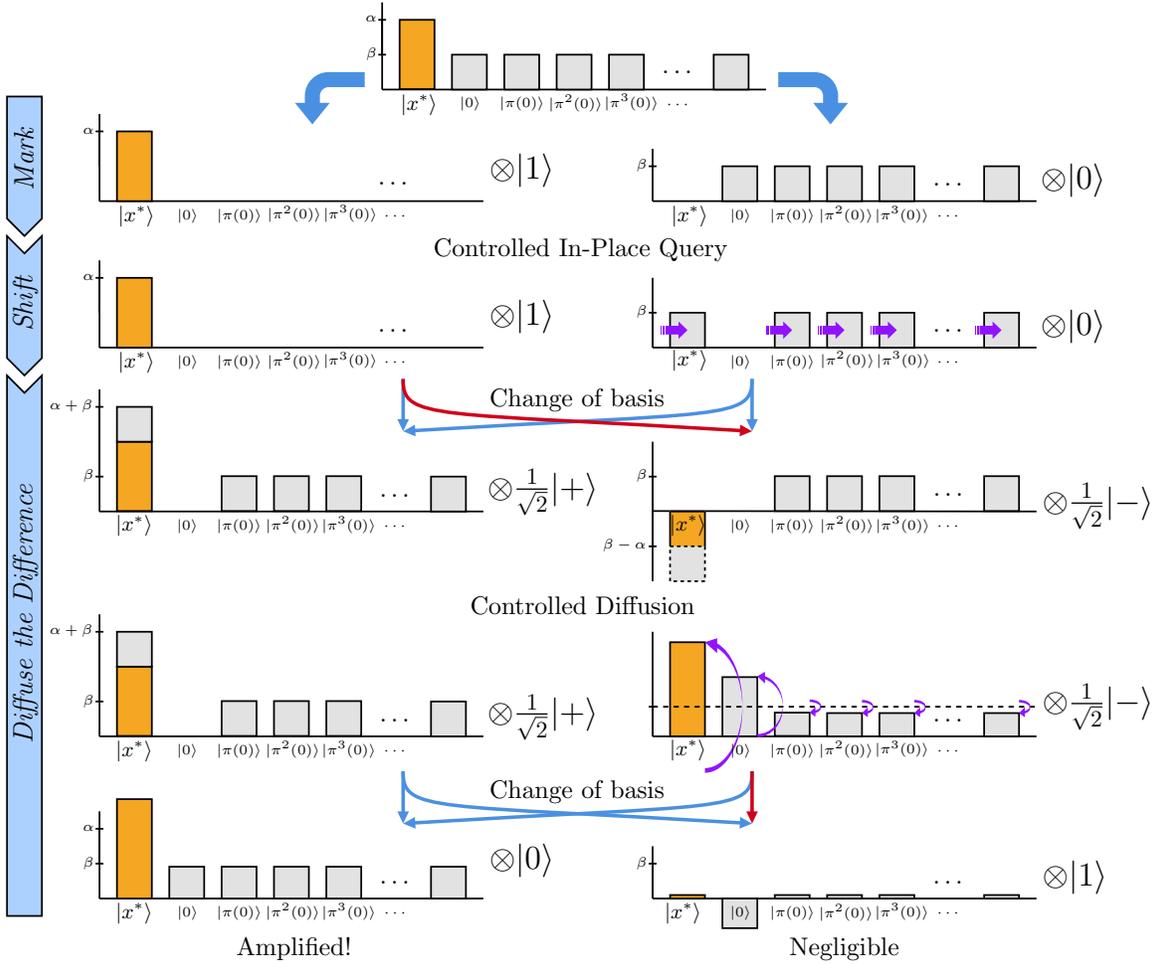}}
    	\caption{
    		(Color)
    		Illustration of how amplitudes change in each iteration of the algorithm.
    		Register $\regB$ is left implicit (note it is unentangled with $\regA$ and $\regC$ by the end of the \textit{Shift} step).
    		Each iteration begins with the nearly uniform superposition from \cref{eq:psitForm}.
    		The \textit{Mark} step queries $\pi$ and creates a marked branch and an unmarked branch, illustrated in two columns.
    		The \textit{Shift} step makes a query in only the unmarked branch,
    		shifting amplitude onto $\ket{x^*}$.
    		The \textit{Diffuse the Difference} step is controlled on $\ket{-}$,
    		so we first rewrite the basis of $\regC$, rearranging amplitudes accordingly.
    		Red and blue arrows indicate positive and negative contributions.
    		The diffusion operator reflects all amplitudes about their mean.
    		A final change of basis leaves a state almost entirely entangled with $\ket{0}$ and with
    		increased amplitude on $x^*$.
    	}
    	\label{fig:algFigure}
    \end{figure}

    \clearpage

    Finally, for the sake of analysis, we choose to measure $\regC$ and abort if $\ket{1}$ is observed.
    We will handle the failure case later.
    For now, we postselect on having observed $\ket{0}$.
    This results in the final (normalized) state
    \begin{align*}
        \ket{\psi_t} &=
        \sqrt{\frac{N}{N-1}} \left[ \left(\beta + \alpha - \frac{\beta+\alpha}{N}\right)\ket{x^*}
        + \sum_{i \in [N] \setminus \set{x^*}} \left(\beta - \frac{\beta+\alpha}{N}\right) \ket{i} \right]
        \ket{0^n}\ket{0}
        \\
        & = \left[
        \sqrt{\frac{N-1}{N}} \left(\beta+\alpha\right) \ket{x^*}
        + \sum_{i \in [N] \setminus \set{x^*}} \left(  \sqrt{\frac{N-1}{N}} \beta - \frac{1}{\sqrt{N}\sqrt{N-1}} \alpha \right) \ket{i}
        \right] \ket{0^n}\ket{0}.
    \end{align*}
    at the end of the $t$-th iteration.
    This state has the form we claimed, with
    \[
        \alpha_{t} = \sqrt{\frac{N-1}{N}} \left(\beta_{t-1}+\alpha_{t-1}\right)
        \quad\text{and}\quad
        \beta_{t} = \sqrt{\frac{N-1}{N}} \beta_{t-1} - \frac{1}{\sqrt{N}\sqrt{N-1}} \alpha_{t-1},
    \]
    concluding our induction.

    The above recurrence lets us write a closed form for $\alpha_t$ and $\beta_t$:
    \begin{align*}
        \begin{bmatrix}
            \alpha_t \\ \beta_t
        \end{bmatrix}
        =
        \begin{bmatrix}
            \sqrt{\frac{N-1}{N}} & \sqrt{\frac{N-1}{N}} \\
            \frac{-1}{\sqrt{N}\sqrt{N-1}} & \sqrt{\frac{N-1}{N}}
        \end{bmatrix}
        \begin{bmatrix}
            \alpha_{t-1} \\ \beta_{t-1}
        \end{bmatrix}
        =
        \begin{bmatrix}
            \sqrt{\frac{N-1}{N}} & \sqrt{\frac{N-1}{N}} \\
            \frac{-1}{\sqrt{N}\sqrt{N-1}} & \sqrt{\frac{N-1}{N}}
        \end{bmatrix}^t
        \begin{bmatrix}
            \frac{1}{\sqrt{N}} \\ \frac{1}{\sqrt{N}}
        \end{bmatrix} .
    \end{align*}
    For a diagonalizable matrix $M=ADA^{-1}$, we know $M^t = AD^tA^{-1}$, so
    we can diagonalize the above matrix to find
    \[
        \alpha_t = \frac{1}{\sqrt{N^{t+1}}} \frac{1}{2i} \left[\left(\sqrt{N-1} + i\right)^{t+1} - \left(\sqrt{N-1} - i\right)^{t+1}\right] .
    \]
    Rewriting the expression in polar form, this is equivalent to
    \[
        \alpha_t = \frac{1}{\sqrt{N^{t+1}}} \frac{1}{2i} \left[\sqrt{N^{t+1}} e^{i\paren{t+1}\theta} - \sqrt{N^{t+1}} e^{-i\paren{t+1}\theta}\right] \quad\text{for}\quad \theta=\arctan{\paren{\frac{1}{\sqrt{N-1}}}} .
    \]
    Finally, the identity $\frac{z-\overline{z}}{2i} = \operatorname{Im}\paren{z} = \sin{\paren{\phi}}$ for $z = e^{i \phi}$
    yields
    \[
        \alpha_t = \sin{\left[ \paren{t+1} \arctan{\left(\frac{1}{\sqrt{N-1}}\right)} \right]} .
    \]

    We want to find the value of $t$ that maximizes $\alpha_t$.
    Setting
    \[
        t^*=\frac{\frac{\pi}{2}}{\arctan{\paren{\frac{1}{\sqrt{N-1}}}}}-1
    \]
    achieves $\alpha_{t^*}=1$.
    The series expansion of this formula shows $t^*$ is asymptotically $\frac{\pi}{2}\sqrt{N} + \bO{1}$, as desired.
    However, $t$ must be an integer, so we set the number of iterations to $T=\floor{t^*}$.
    Observe that
    $\sin{\paren{x}}$ increases as $x$ approaches $\frac{\pi}{2}$,
    so it is sufficient to lower bound $\alpha_{t^*-1} \le \alpha_{T}$.
    Substituting and then simplifying, we find
    \[
        \alpha_{t^*-1}
        = \sin{\left[\frac{\pi}{2} -  \arctan{\left(\frac{1}{\sqrt{N-1}}\right)} \right]}
        = \cos{\left[\arctan{\left(\frac{1}{\sqrt{N-1}}\right)}\right]}
        = \sqrt{1 - \frac{1}{N}} .
    \]
    So, given the algorithm never terminates early, it outputs $\ket{x^*}$ with probability at least $\abs{\alpha_T}^2 \geq 1-1/N$.

    Finally, we handle the possibility of the algorithm terminating early.
    In each iteration, given $\ket{\psi'''_{t-1}}$, the probability of measuring $\ket{1}$ is
    $\frac{\alpha^2 + \paren{N-1}\beta^2}{N} = \frac{1}{N}$.
    Therefore, in $T=\bOnoparen(\sqrt{N})$ iterations, the probability of aborting is at most a negligible $T/N=\bO{1/\sqrt{N}}$.

    Overall, we have that our algorithm aborts with probability at most $\bO{1/\sqrt{N}}$,
    while if it does not abort, then it fails to find $\ket{x^*}$ with probability at most $\bO{1/N}$.
    We conclude that with $T= \frac{\pi}{2}\sqrt{N}+\bO{1}$ queries to $\Ppi$, we can solve \permInvers{} with probability $1 - \bO{1/\sqrt{N}}$.
\end{proof}

\section{Simulating Other Oracles}\label{sec:otherOracles}
This section tightly characterizes the relationship between \XOR{} oracles and in-place oracles.

Using Grover's algorithm, $\bOnoparen(\sqrt{N})$ queries to an \XOR{} oracle $\Spi$ are sufficient to simulate the \XOR{} oracle for the inverse, $\Spiinv$.
In fact, using a zero-error variant of Grover's algorithm \cite{Hoyer2000-exactGrover}, this simulation is exact.
Then, since one query to each of $\Spi$ and $\Spiinv$ are sufficient to implement
either the in-place oracle $\Ppi$ or the in-place oracle for the inverse $\Ppiinv$ (see \cref{sec:oracles}),
it follows that $\bOnoparen(\sqrt{N})$ queries to the \XOR{} oracle $\Spi$ are sufficient to exactly implement any of the oracles $\Spiinv,\Ppi$, or $\Ppiinv$.
Moreover, these algorithms are optimal \cite{Kashefi-KKVB02comparisonQuantumOracles}.

Here, using our new algorithm for \permInvers{} from \cref{sec:permInversion},
we give analogous results for in-place oracles.
We show that $\bOnoparen(\sqrt{N})$ queries to an in-place oracle $\Ppi$ are sufficient to approximately simulate queries to $\Ppiinv$, $\Spi$, or $\Spiinv$.
Our simulations are approximate only because our algorithm succeeds with probability inverse-exponentially close to 1 but not zero-error.
Still, these approximate constructions are more than sufficient to emulate any $\BQP{}$ algorithm querying $\Ppiinv,\Spi,$ or $\Spiinv$.
Then, we prove that our simulations are optimal.
This establishes the first task known to require
fewer queries to an \XOR{} oracle than an in-place oracle, namely the approximate unitary-implementation problem of simulating an \XOR{} oracle.

The query complexity of unitary-implementation problems,
such as the task of simulating the unitary $\Spiinv$,
has been formalized by Belovs \cite{Belovs15variationsOnQuantum},
and we discuss it further in \cref{sec:subspace}.
We emphasize that we indeed give a unitary construction,
rather than incoherently solving the problem on basis states.
Although our implementation of a controlled in-place query from \cref{sec:controlledinplace} introduces a query-independent garbage register $\ket{\pi(0)}$,
we will show how to use another application of the $\permInvers{}$ algorithm to erase it.

Most of our work is a construction for $\Spiinv$ in \cref{lem:PermImplementSinv}.
Then, we get $\Spi$ in \cref{thm:PermImplementXor_upper}, and finally $\Ppiinv$ in \cref{corr:PermImplementPerminv}.
Later, we show matching lower bounds for these results.

\begin{lemma}\label{lem:PermImplementSinv}
    For a permutation $\pi$ on $[N]$, $\bOnoparen(\sqrt{N})$ queries to
    the in-place oracle $\Ppi$ are sufficient to approximately implement
    the \XOR{} oracle for the inverse function, $\Spiinv$.
\end{lemma}
\begin{proof}
    We will use our algorithm for \permInvers{} to construct a unitary that maps the state $\ket{x}\ket{y}$ inverse-exponentially close to the state $\ket{x}\ket{y \oplus \pi^{-1}(x)}$ for all $x,y \in [N]$.
    By linearity, this unitary will approximate the unitary $\Spi$.

    First, we modify our \permInvers{} algorithm.
    It is straightforward to modify the algorithm to invert $\pi$ on an arbitrary input instead of $0$.
    Next, by the quantum union bound \cite{Gao15unionBound, ODonnell-OV22quantumunionbound},
    omitting the intermediate measurements of the optional \textit{Measure} step
    results in a final state with trace distance at most $\bO{N^{-1/4}}$ from the desired output.
    With these modifications, and recalling that our construction of $\cont\Ppi$ from \cref{sec:controlledinplace} requires a catalyst $\ket{\pi(0)}$,
    our algorithm becomes a unitary $A_\pi$ that approximately maps
    \begin{equation}\label{eq:Api}
           A_\pi \ket{x}\ket{0}_{\regA}\ket{0}_{\regB}\ket{0}_{\regC}\ket{\pi(0)} \approx \ket{x}\ket{\pi^{-1}(x)}_{\regA}\ket{0}_{\regB}\ket{0}_{\regC}\ket{\pi(0)}.
    \end{equation}
    Henceforth we omit the auxiliary registers $\regB$ and $\regC$.

    Now we show how to erase a register containing $\ket{\pi(0)}$ using $\bOnoparen(\sqrt{N})$ in-place queries to $\Ppi$.
    A gate-by-gate inversion of the circuit for $A_\pi$ performs \cref{eq:Api} in the reverse direction but makes queries to $\paren{\Ppi}^{-1} = \Ppiinv$.
    If we replace these queries to $\Ppiinv$ with queries to $\Ppi$, we get a circuit that we label $B_\pi$ which makes $\bOnoparen(\sqrt{N})$ queries to $\Ppi$ such that
    \begin{align}\label{eq:Bpi}
        B_\pi \ket{x}\ket{\pi(x)}\ket{\pi^{-1}(0)} \approx \ket{x}\ket{0}\ket{\pi^{-1}(0)}
    \end{align}
    with the same approximation as in \cref{eq:Api}.
    Starting with $\ket{0}\ket{0}\ket{\pi(0)}$,
    we can combine these circuits to approximate the following steps:
    \begin{align*}
        \ket{0}\ket{0}\ket{\pi(0)}
        \xrightarrow{A_\pi} \ket{0}\ket{\pi^{-1}(0)}\ket{\pi(0)}
        \xrightarrow{\SWAP} \ket{0}\ket{\pi(0)}\ket{\pi^{-1}(0)}
        \xrightarrow{B_\pi} \ket{0}\ket{0}\ket{\pi^{-1}(0)}
        \xrightarrow{\Ppi} \ket{0}\ket{0}\ket{0}
        .
    \end{align*}

    By querying $\Ppi$ to generate the catalyst $\ket{\pi(0)}$ before applying $A_\pi$ and then using the above procedure to erase $\ket{\pi(0)}$ afterward, we define a unitary that approximately converts $\ket{x}\ket{0}$ to $\ket{x}\ket{\pi^{-1}(x)}$ for all $x \in [N]$, ignoring auxiliary qubits that start and end as $\ket{0}$.

    From here it is straightforward to simulate $\Spiinv$.
    Given any state $\ket{x}\ket{y}$, apply the full unitary to approximately get $\ket{x}\ket{y}\ket{\pi^{-1}(x)}$,
    \XOR{} one register into another and then query $\Ppi$ to get
    $\ket{x}\ket{y \oplus \pi^{-1}(x)}\ket{x}$,
    and finally \XOR{} the registers to reset the auxiliary register,
    leaving us approximately with the state $\ket{x}\ket{y \oplus \pi^{-1}(x)}$, as desired.
\end{proof}

Next, by a straightforward extension of the previous proof, we show how to implement the \XOR{} oracle $\Spi$ using an in-place oracle $\Ppi$.

\begin{lemma}\label{thm:PermImplementXor_upper}
    For a permutation $\pi$ on $[N]$,
    $\bOnoparen(\sqrt{N})$ queries to the in-place oracle $\Ppi$ are sufficient to approximately implement the \XOR{} oracle $\Spi$.
\end{lemma}

\begin{proof}
    The proof of \cref{lem:PermImplementSinv} gives a unitary making $\bOnoparen(\sqrt{N})$ queries to $\Ppi$ to approximate $\Spiinv$.
    Because $\Spiinv$ is self-inverse, if we invert this circuit gate-by-gate, we get a circuit making $\bOnoparen(\sqrt{N})$ queries to $\Ppiinv$ that still approximates $\Spiinv$.
    Replacing queries to $\Ppiinv$ with queries to $\Ppi$ yields a circuit making $\bOnoparen(\sqrt{N})$ queries to $\Ppi$ that now approximates $\Spi$.
\end{proof}

As stated in \cref{sec:oracles}, one query to each of $\Spi$ and $\Spiinv$ is sufficient to implement $\Ppiinv$, implying the following corollary.

\begin{corollary}\label{corr:PermImplementPerminv}
    For a permutation $\pi$ on $[N]$,
    $\bOnoparen(\sqrt{N})$ queries to the in-place oracle $\Ppi$ are sufficient to approximately implement the in-place oracle to the inverse, $\Ppiinv$.
\end{corollary}

Next, we prove the above simulations of $\Spiinv$, $\Spi$, and $\Ppiinv$ are optimal.
The $\bOmnoparen(\sqrt{N})$ query lower bound on \permInvers{} by \cite{Fefferman-FK18-QCMA} immediately implies the lower bounds for implementing $\Spiinv$ and $\Ppiinv$.
We prove the remaining lower bound,
that at least $\bOmnoparen(\sqrt{N})$ queries to $\Ppi$ are required to simulate $\Spi$, stated in \cref{thm:inPlaceSimulatingXOR_lower}, below.
Our proof is similar to the proof by \cite{Kashefi-KKVB02comparisonQuantumOracles} that $\bOmnoparen(\sqrt{N})$ queries to $\Spi$ are necessary to simulate $\Ppi$,
but it relies on the later result of \cite{Fefferman-FK18-QCMA}.

\begin{proof}[Proof of \cref{thm:inPlaceSimulatingXOR_lower}]
    Suppose for the sake of contradiction that a circuit closely approximates $\Spi$ with $\lonoparen(\sqrt{N})$ queries to $\Ppi$.
    Because $\Spi$ is self-inverse, if we invert this circuit gate-by-gate,
    then we get a circuit making $\lonoparen(\sqrt{N})$ queries to $\Ppiinv$ that still approximates $\Spi$.
    Replacing queries to $\Ppiinv$ with queries to $\Ppi$ yields a circuit making $\lonoparen(\sqrt{N})$ queries to $\Ppi$ that approximates $\Spiinv$,
    violating the lower bound of \cite{Fefferman-FK18-QCMA}.
    We conclude that every constant-approximation of $\Spi$ requires at least $\bOmnoparen(\sqrt{N})$ queries to $\Ppi$.
\end{proof}

Because $\Spi$ can clearly simulate $\Spi$ with one query, the above lower bound makes the approximate unitary-implementation problem of simulating an \XOR{} query to $\pi$
the first task known to require more in-place queries than \XOR{} queries.
We improve on this result in the next section.

Combining these upper and lower bounds, including the results of \cite{Kashefi-KKVB02comparisonQuantumOracles,Fefferman-FK18-QCMA},
the relationship between in-place and \XOR{} oracles is neatly summarized in \cref{cor:summaryRelationships}, which we restate.

\summary*

\section{A Subspace-Conversion Separation}\label{sec:subspace}
In this section,
we give a subspace-conversion problem which can be solved with 1 query to an \XOR{} oracle but requires $\bOmnoparen(\sqrt{N})$ queries to an in-place oracle.
This improves on the unitary-implementation separation from the previous section.
Together, these are the first problems shown to require more in-place queries than \XOR{} queries.
We begin by elaborating on unitary-implementation, subspace-conversion, and other types of problems.

Problems in query complexity traditionally have the goal of evaluating some function of the input oracle.\footnote{
    In query complexity,
    the oracle itself is usually considered the input, so a problem is a function of the oracle.
    Sometimes the oracle is also referred to as a function, since it maps queries to answers, and it is important to distinguish between the two.
    }
These are ``function evaluation'' problems, whether they are search problems or decision problems.
State-generation problems were introduced in the context of query complexity by \cite{AMRR11-symmetryStateGeneration}
and require preparing some oracle-dependent state from an initial all-zeroes state.
These problems generalize function evaluation since
function evaluation can be reduced to the state-generation task of preparing a state which contains the desired answer.
State-generation was generalized to state-conversion problems by Lee, Mittal, Reichardt, \v{S}palek, and Szegedy \cite{LMRSS11-stateConversion},
where the problem is to accept an oracle-dependent input state and produce some oracle-dependent output state.
Later,
\cite{Belovs2023onewayticketlasvegas} generalized these to subspace-conversion problems,
which require simultaneously performing state-conversion on many states.
Finally, as we discuss shortly, one can define unitary-implementation problems.

Algorithms for state-generation, state-conversion, and subspace-conversion problems
can be defined to require exact or approximate solutions
and to require coherent or non-coherent solutions.
As an example,
suppose a state-generation problem requires, on input $x$, outputting the state $\ket{\psi_x}$.
Then here, the \emph{non-coherent} state-generation problem would be satisfied by any output of the form
$\ket{\psi_x}\ket{\text{gar}_x}$, where garbage is even allowed to depend on the input.
This garbage may prevent quantum interference if an algorithm for this problem is used as a subroutine in some larger algorithm.
The definition of coherent versus non-coherent does not necessarily mean anything about how a solution may be implemented.

In this work we restrict our consideration to coherent problems.
Then, there is a hierarchy of problems: decision problems being least general to subspace-conversion problems being most general
\cite[Fig.\ 1]{LMRSS11-stateConversion}, \cite[Fig.\ 2]{Belovs2023onewayticketlasvegas}.
Additionally, coherent subspace-conversion problems
can be generalized to unitary-implementation problems \cite{Belovs15variationsOnQuantum},
for which a solution must act on all possible input states rather than some strict subspace.
Then, the hierarchy is from decision problems to unitary-implementation problems.
More precisely, a decision problem can be embedded into a search problem, which can be embedded into a state-generation, then state-conversion,
then subspace-conversion, and finally, unitary-implementation problem.
Observe that a separation of the \XOR{} and in-place query complexities of
a decision problem would imply separations for all of the more general types of problems.
Therefore, a separation with one type of problem can be seen as formally stronger or weaker than with another.

To improve our unitary-implementation separation to a coherent subspace-conversion separation,
we study an analogue of the \indexErasure{} problem,
which is known to require $\bTnoparen(\sqrt{N})$ queries to an \XOR{} oracle but just 1 query to an in-place oracle \cite{AMRR11-symmetryStateGeneration}.
Recall that \indexErasure{} is the state-generation problem of preparing the state $\frac{1}{\sqrt{N}}\sum_{x\in[N]} \ket{f(x)}$,
which is a special-case of the state-conversion problem converting $\frac{1}{\sqrt{N}}\sum_{x\in[N]}\ket{x}\ket{f(x)}$ to $\frac{1}{\sqrt{N}}\sum_{x\in[N]}\ket{f(x)}$.
Flipping perspectives, while it is hard to erase the index register using an \XOR{} oracle,
it appears hard to erase the output register using an in-place oracle.

A state-generation ``output erasure'' problem similar to \indexErasure{} does not seem appropriate since it is trivial to generate the state $\frac{1}{\sqrt{N}}\sum\ket{x}$,
and we do not know how to prove a lower bound for the state-conversion version.
So,
we define the coherent subspace-conversion problem \functionErasure{} (\cref{def:funcErasure}),
which is to convert states of the form $\sum \alpha_x \ket{x}\ket{f(x)}$ to $\sum_x \alpha_x\ket{x}$.
It is straightforward to solve \functionErasure{} with 1 \XOR{} query.
Below, we prove \cref{thm:inPlaceAndFE}, that $\bTnoparen(\sqrt{N})$ in-place queries to $\pi$ are necessary and sufficient to solve $\functionErasure{}$ for $\pi$.
This gives a subspace-conversion separation.

\begin{proof}[Proof of \cref{thm:inPlaceAndFE}]
    The upper bound follows from our simulation of one \XOR{} query with $\bOnoparen(\sqrt{N})$ in-place queries and the fact that \functionErasure{} can be solved with one \XOR{} query.

    To prove the lower bound, we first argue that
    an \XOR{} query can be simulated by one call to \functionErasure{} and one additional in-place query.
    To see this, observe that given any state of the form $\ket{x}\ket{y}$, we can \XOR{} $x$ into an auxiliary register,
    then query an in-place oracle to $\pi$ to get $\ket{x}\ket{y}\ket{\pi(x)}$,
    \XOR{} to get $\ket{x}\ket{b \oplus \pi(x)}\ket{\pi(x)}$,
    and finally use an algorithm for \functionErasure{} to erase the auxiliary register,
    leaving $\ket{x}\ket{y \oplus \pi(x)}$, as desired.
    By linearity, this procedure simulates the unitary $\Spi$.

    By \cref{thm:inPlaceSimulatingXOR_lower},
    any algorithm simulating an \XOR{} query using an in-place oracle requires
    $\bOmnoparen(\sqrt{N})$ queries.
    We conclude that any algorithm for \functionErasure{} using in-place queries,
    as in the above simulation of an \XOR{} query,
    must use at least $\bOmnoparen(\sqrt{N})$ queries,
\end{proof}
\section{Lower Bounds}\label{sec:lowerbounds}
In this section, we consider avenues for improving our separations with \XOR{} oracles outperforming in-place oracles
to demonstrate a decision-problem separation.

In \cref{sec:lowerbounds-conventionalTechniques}, we explain the limitations of conventional lower bound techniques for showing
that fewer \XOR{} queries are required for a task than in-place queries.
In \cref{sec:candidateProblem}, we introduce a candidate decision problem which we conjecture exhibits such a separation.
Then in \cref{sec:techAndCondForSeparation},
we explore recently developed tools for proving lower bounds for arbitrary oracles, including in-place oracles.
We develop
exact conditions for a decision problem to exhibit a separation.
Further details are given in \cref{app:adv}.

\subsection{Conventional Lower Bound Techniques}\label{sec:lowerbounds-conventionalTechniques}
In pursuit of proving a decision-problem separation with \XOR{} oracles outperforming in-place oracles,
we begin by considering standard tools.
Two techniques have dominated quantum query complexity: the polynomial method and the (basic) adversary method.
See the thesis of Belovs \cite[Chap.\ 3]{Belovs2014thesis} for an excellent survey of these tools.
Unfortunately, we find that these two methods are insufficient for proving the desired separation.

\paragraph{The Polynomial Method}
The polynomial method as applied to quantum computation was developed by Beals, Buhrman, Cleve, Mosca, and de~Wolf \cite{BBCMdW01-polynomialMethod}.
Briefly, the core idea behind the method is that every final amplitude of a quantum query algorithm
can be written as a polynomial with degree equal to the number of queries
and over variables corresponding to the entries in the oracle unitary.
So, if a query algorithm exists that maps queries to particular probabilities,
then a polynomial of a certain degree exists that maps the associated variables to the same values.
If one proves
a lower bound on the necessary degree of such a polynomial,
then one proves
a lower bound on the query complexity.

The standard method of converting a query algorithm
into a polynomial can be applied when querying an in-place oracle.
For the unitary encoding an in-place oracle $\Ppi$,
an entry at column $x$ and row $z$ is 1 if $\pi(x)=z$ and 0 otherwise.
So, the variables used in the polynomial are equivalent to indicator variables
$b_{z,x}$ for the event $\pi(x)=z$.
This same set of variables can be used when applying the polynomial method
to an algorithm querying an \XOR{} oracle $\Spi$.
More specifically, an entry at column $\ket{x}\ket{a}$ and row $\ket{x}\ket{a\oplus z}$ in
the unitary $\Spi$
is equivalent to $b_{z,x}$, the indicator for the event $\pi(x)=z$.
Moreover, when using these indicator variables,
the same argument that the degree of the polynomials increments by at most one for each query
applies identically in both cases.
So, an algorithm making $t$ queries to $\Ppi$ and an algorithm making $t$ queries to $\Spi$
can both be associated with polynomials
of the same maximum degree over the same variables.
Proving a lower bound on the degree of such a polynomial implies a lower bound for both query algorithms.
For this reason, any lower bound for in-place oracles proved using the polynomial method implies an identical lower bound for \XOR{} oracles, preventing the (basic/standard) polynomial method from helping us prove
a separation requiring fewer \XOR{} queries than in-place queries.

Surprisingly though, separations in the other direction are possible using the polynomial method,
if applied indirectly.
For the \setComp{} problem, an approximate version of \setEquality{} or \collision{},
Aaronson \cite{Aaronson02collisionProblem} proved an exponential lower bound on the necessary number of \XOR{} queries using the polynomial method
and then gave an algorithm using $\bO{1}$ in-place queries.
The reason this avoids the barrier described above is subtle.
\setComp{} involves distinguishing a nearly one-to-one function from a nearly two-to-one function.
A query lower bound was proved for the formally easier problem of distinguishing a one-to-one function from a many-to-one function,
which implies a lower bound on \setComp{}.
But, although queries in \setComp{} are structured such that queries to a two-to-one function are still injective,
queries to a many-to-one function are not.
So while the \XOR{} query lower bound for the easier problem implies an \XOR{} query lower bound for \setComp{},
the easier problem is undefined for in-place oracles!
To be clear, this does not contradict our previous conclusions about the polynomial method.
Here, the polynomial method was used to prove a lower bound on \XOR{} queries,
while our conclusion above was for bounds proved on in-place queries.
In particular,
proving the bound for \setComp{} indirectly, via the easier problem,
circumvents making the above argument starting with the bound on \XOR{} queries.
Unfortunately,
this indirect trick only works in one direction, as
any function that can be implemented in the in-place model can be implemented in the \XOR{} model.

\paragraph{The Adversary Method}
The quantum adversary method, in its many forms, is characterized by its use of a ``progress measure''.
Informally, there is
a maximum amount of progress that can be made with one query
and a minimum amount of progress necessary to solve a problem,
implying a lower bound on the number of queries.

The original, ``basic'' unweighted adversary method was developed by Ambainis \cite{Ambainis02adversaryMethod} for \XOR{} query lower bounds.
Fefferman and Kimmel \cite{Fefferman-FK18-QCMA} later reproved the adversary theorem for in-place oracles.
In fact, they proved the theorem with the exact same parameters as the original.
This implied the $\bOmnoparen(\sqrt{N})$ query lower bound for \permInvers{}
with \XOR{} oracles immediately extends to in-place oracles,
and in fact that any lower bound proved using Ambainis's original method extends to in-placed oracles.
Conversely, any lower bound on the number of in-place queries proved using the basic adversary method
immediately implies the same lower bound for queries to \XOR{} oracles.
Therefore, this (basic) technique is unable to prove a separation.

One might wonder whether, even though this particular adversary theorem cannot help, the underlying technique can.
On one hand, in-place oracles deal only with permutations rather than general functions,
so there may be some way to develop another relatively simple tool for lower bounds with in-place oracles.
On the other hand, any method which simply follows the same approach of dividing total progress by the maximum progress per query
will face a barrier for some problems.
For example, Grover's algorithm with \XOR{} oracles makes the maximum amount of progress in its first query,
mapping $\frac{1}{\sqrt{N}} \sum \ket{x}\ket{0}$ to $\frac{1}{\sqrt{N}}\sum\ket{x}\ket{\pi(x)}$.
But that first query can be simulated identically using an in-place oracle.
So, the maximum amount of progress per query is the same for both oracles.
It seems that progress with in-place oracles only slows in successive queries,
when the inability to uncompute garbage prevents interference and entanglement.

\subsection{A Candidate Decision Problem}\label{sec:candidateProblem}
In this section, we introduce a decision problem called $\perminvgarb$ which can be solved with $\bTnoparen(\sqrt{N})$ \XOR{} queries
and which we conjecture requires $\bOm{N}$ in-place queries.

Earlier, we showed that in-place query algorithms can achieve the same query complexity as \XOR{} oracles for \permInvers{}.
As noted in \cref{footnote:injections},
our algorithm appears to crucially rely on the fact that it is inverting a permutation rather than an injection.
The algorithm uses the image of the permutation from one iteration as the input in the next.
Now that our goal is to find a problem for which in-place queries are less useful than \XOR{} queries, we leverage this limitation. (Below, $S_i$ is the symmetric group of degree $i$.)

\begin{definition}[Promised Permutation Inversion]\label{def:permInvGarb}
    Given query access to a permutation $f$ on $[N^2]=\set{0,\dots,N^2-1}$,
    the decision problem
    $\perminvgarb : S_{N^2}\to\zo$
    is defined by
    \[
        \perminvgarb(f)=\begin{cases}
            1&\text{if $f^{-1}(0)\le N$, and}\\
            0&\text{otherwise.}
    \end{cases}\]
\end{definition}

In effect, this problem embeds an injection from $\left[N\right] \rightarrow \left[N^2\right]$ into a bijection on $\left[N^2\right]$,
with the promise that an algorithm only needs to search over $[N]$.
This problem is inspired by a candidate proposed by Aaronson \cite{Aaronson21-querySurvey}
which was a version of Simon's problem with garbage appended to each query.
When querying an \XOR{} oracle, it is easy to copy the desired part of any answer and then uncompute with an additional query,
allowing the garbage to be ignored.
In contrast, it is unclear how to uncompute or erase the garbage with an in-place oracle, which would prevent interference.
Here, instead of Simon's problem we focus on \permInvers{},
and we formalize the idea of appending garbage as embedding an injection into a bijection.

\begin{lemma}\label{lem:perminvgarbXor}
    $\perminvgarb$ can be decided with at most $\bTnoparen(\sqrt{N})$ \XOR{} queries.
\end{lemma}
\begin{proof}
    The algorithm is simply to perform Grover's algorithm over $[N]$.
    If a pre-image of $0$ is found in $[N]$, then accept, and otherwise reject.

    In more detail,
    in each iteration of the algorithm, an element $x$ is marked if $f(x)=0$.
    This is straightforward because for any $x\in [N]$, a query can be made to produce $\ket{x}\ket{f(x)}$ and a phase can be conditionally applied.
    Then since \XOR{} oracles are self-inverse, it is easy to make another query to uncompute $f(x)$, leaving behind only the phase on $\ket{x}$.
\end{proof}

It is unclear how to solve \perminvgarb{} as efficiently as the above algorithm when using in-place queries.
Simply querying $\sum \ket{x} \mapsto \sum\ket{f(x)}$
would be useless.
One can instead consider algorithms that involve mapping $\ket{x}\ket{0}\mapsto \ket{x}\ket{f(x)}$.
Any such query $x\in [N]$ will lead to an unknown element $f(x)\in [N^2]$.
Since $f(x)$ may not be in $[N]$, this (a priori) unknown element seems useless for finding the pre-image $f\inv(0)\in [N]$.
Moreover, in-place oracles are not self-inverse and do not readily allow uncomputing queries.
So the image is both useless to keep around and not readily uncomputable using in-place queries.
We conjecture this task as a candidate for which \XOR{} oracles outperform in-place oracles.

\begin{conjecture}\label{conj:perminvgarb}
    $\perminvgarb$ requires at least $\bOm{N}$ queries to an in-place oracle.
\end{conjecture}

Note that even a classical algorithm can solve the problem with $N$ queries by simply querying every element of $[N]$.  Also note that while an exponential query separation is possible for Simon's with garbage, the largest separation possible with \perminvgarb{} is polynomial. We hope that the structure of the problem makes a separation more tractable.

\subsection{Sketch of Techniques for a Decision Problem Separation}
\label{sec:techAndCondForSeparation}
Here, we briefly
we explore applying a recent version of the quantum adversary bound to prove the desired
decision-problem separation.
A full exposition is given in \cref{app:adv}.

Quantum query complexity can be characterized by the adversary method. This method has been
used to develop several different adversary bounds or adversary theorems in different contexts.
For example, prior work derived adversary bounds in the \XOR{} oracle model.
In general, an adversary
bound for a decision problem $\phi:D\to\zo$ is an optimization problem such that the optimum is
a lower bound on the query complexity.
Belovs and Yolcu \cite{Belovs2023onewayticketlasvegas}
recently developed a new version of the adversary bound that
applies to arbitrary linear transformations.
In fact, \cite[Section 10]{Belovs2023onewayticketlasvegas} specifically observed this includes
in-place oracles in addition to \XOR{} oracles.
Moreover, the bound of \cite{Belovs2023onewayticketlasvegas} is tight,
meaning the optimum value of the optimization problem
corresponds to the optimum query complexity and vice versa.

One caveat is that the lower bound of \cite{Belovs2023onewayticketlasvegas} is for \textit{Las Vegas} query complexity,
a quantum analog of the expected number of queries needed for a zero-error algorithm,
in contrast to the usual notion of bounded-error complexity.
So, our results in this section are primarily focused on Las Vegas complexity.
But, for the special case of \perminvgarb{}, we are able to extend the analysis to bounded-error complexity.

The optimization problem in the adversary bound developed by \cite{Belovs2023onewayticketlasvegas}
is specifically an optimization over \textit{adversary matrices} $\Gamma$.
The optimal choice of adversary matrix then corresponds to the optimal query algorithm.
In other versions of the adversary method,
adversary matrices have been restricted to nonnegative values (the positive weight method)
or to general real numbers (the negative weights method).
For a decision problem $\phi:D\to\zo$,
previous methods have nearly always restricted $\Gamma$ such that
an entry $\Gamma[f,g]$ indexed by problem instances $f$ and $g$
satisfies that if $\phi(f)=\phi(g)$, then $\Gamma[f,g]=0$.
But, one feature of this new version of the adversary method is that it removes that restriction:
we are free to assign nonzero values to \textit{all} entries of $\Gamma$.

We call these matrices, with nonzero entries corresponding to problem instances with the same answer,
\textit{extended} adversary matrices.
We show that, just as negative-weight adversary matrices are necessary to prove tight lower bounds
for certain problems,
these ``extended'' adversary matrices
are necessary to prove the desired decision-problem separation
with \XOR{} oracles outperforming in-place oracles.
In other words, if we use only tools from the negative-weight adversary bound to construct adversary matrices $\Gamma$,
then we cannot prove our desired query separation.

\begin{theorem*}[Informal statement of \cref{cor:needextend}]\label{inf:lv}
    For a decision problem $\phi:D\to \zo$,
    the
    Las Vegas query complexity using \XOR{} oracles
    is asymptotically less than the Las Vegas query complexity using in-place oracles
    if and only if
    optimizing over extended adversary matrices witnesses it.
\end{theorem*}

Again, the above statement is in terms of Las Vegas complexity
instead of the more typical bounded-error complexity.
But, for our candidate problem \perminvgarb{} introduced in the previous section,
we are able to extend the statement to bounded-error complexity

\begin{theorem*}[Informal statement of \cref{claim:perminvgarbExtended}]\label{inf:perminv}
    For the decision problem $\perminvgarb{}$,
    the bounded-error query complexity using \XOR{} oracles is asymptotically less than the
    bounded-error query complexity using in-place oracles
    if and only if
    optimizing over extended adversary matrices witnesses it.
\end{theorem*}

See \cref{app:adv} for details.
In sum, we considerably narrow down what techniques could possibly prove an $\bOm{N}$ lower bound
on \perminvgarb{}.
Although we rule out the polynomial and unweighted adversary methods,
the new adversary method of Belovs and Yolcu \cite{Belovs2023onewayticketlasvegas} is tight,
so that if such a lower bound is possible,
then it is witnessed by adversary matrices.
By the above theorem,
we see that any lower bound stronger than $\bOmnoparen(\sqrt{N})$ must use
this new class of \textit{extended} adversary matrices.

\section*{Acknowledgments}
We especially thank John Kallaugher for helpful conversations and mentorship.
RR and JY thank Scott Aaronson for stimulating discussions.
JY thanks Bill Fefferman for suggesting his conjecture on \permInvers{}.

Sandia National Laboratories is a multimission laboratory managed and operated by National
Technology and Engineering Solutions of Sandia, LLC., a wholly owned subsidiary of Honeywell
International, Inc., for the U.S. Department of Energy’s National Nuclear Security Administration
under contract DE-NA-0003525.
This work is supported by a collaboration between the US DOE and other Agencies.
BH and JY acknowledge this work was supported by the U.S. Department of Energy, Office
of Science, Office of Advanced Scientific Computing Research, Accelerated Research in Quantum
Computing, Fundamental Algorithmic Research for Quantum Utility, with support also acknowledged from Fundamental Algorithm Research for Quantum Computing.

JY acknowledges this material is based upon work supported by the U.S. Department of Energy, Office of Science, National Quantum Information Science Research Centers, Quantum Systems Accelerator.
RR was supported by the NSF AI Institute for Foundations of Machine Learning (IFML).

\newpage
\bibliographystyle{alphaurl}
\bibliography{bibliography}
\newpage

\appendix
\crefalias{section}{appendix}

\section{Techniques and Conditions for a Decision Problem Separation}
\label{app:adv}
In \cref{sec:otherOracles,sec:subspace},
we showed there exist unitary-implementation and subspace-conversion problems that
require more in-place queries than \XOR{} queries.
But it remains open whether there is such a separation for a decision problem.
We discussed ideas for demonstrating such a separation in \cref{sec:techAndCondForSeparation},
and we expand on those ideas here.
First, we will explore the recent version of the quantum adversary method developed by Belovs and Yolcu \cite{Belovs2023onewayticketlasvegas}.
Their method is applicable to in-place oracles, although it technically applies to Las Vegas
query complexity instead of the usual notion of bounded-error query complexity.
Second, we will prove
\cref{cor:needextend},
that for any decision problem,
the Las Vegas query complexity using \XOR{} queries is less than the Las Vegas query complexity using in-place oracles
if and only if it is witnessed by a novel type of adversary matrix, which we call ``extended''.
Third, for our candidate problem \perminvgarb{} introduced in \cref{sec:candidateProblem},
we are able to make the extend the statement to bounded-error query complexity,
giving \cref{claim:perminvgarbExtended}.

\subsection{Las Vegas Query Complexity}
Quantum query complexity for Boolean functions is usually described with respect to bounded-error algorithms.
Consider an algorithm $A$ for a decision problem $\phi:D\to\zo$ with the oracle set $\calO =\{O_f\mid f\in D \}$.
This oracle set $\calO$ defines how to map an underlying function $f$ to a particular unitary $O_f$,
such as an \XOR{} oracle $\oracle{S}_f$ or in-place oracle $\oracle{P}_f$.
For an instance $f\in D$, the algorithm $A$ is given query access to $O_f$.
After some number of queries $q_f$,
the algorithm outputs $b\in \zo$.
We say that $A$ is a \textit{bounded-error algorithm} for $\phi$ if for all $f\in D$,
the output $b$ equals $\phi(f)$ with probability at least $2/3$.
The algorithm $A$ has query complexity $q=\max_{f\in D}q_f$.
The \textit{bounded-error query complexity} of the problem $\phi$
is the minimum query complexity across all bounded-error algorithms for $\phi$.
We use $\qbxor(\phi)$
to refer to the bounded-error quantum query complexity of $\phi$ when the functions $f\in D$ are encoded as \XOR{} oracles, \ie{} $\calO_f=\oracle{S}_f$,
and we use $\qbperm(\phi)$ for the query complexity using in-place oracles (``perm'' for ``permutation oracles'').

\textit{Las Vegas query complexity} for quantum computation was introduced by \cite{Belovs2023onewayticketlasvegas}.
Classically, Las Vegas query complexity refers to the expected number of queries made by a
zero-error algorithm,
in contrast to bounded-error (a.k.a.\xspace{} Monte Carlo) algorithms.
The definition of Las Vegas query complexity by \cite{Belovs2023onewayticketlasvegas} is
effectively an extension of the classical definition,
but formalizing it for quantum computation is nuanced.
See \cite{Belovs2023onewayticketlasvegas} for a formal definition and technical details.
Briefly, a Las Vegas quantum query algorithm $B$ for problem $\phi$ has oracle access to $O_f$
but also has the ability to ``skip queries''.
To formalize this in the circuit model, the algorithm
$B$ is given query access to $\tilde O_f=O_f\oplus I$ instead of $O_f$.
Then, the \textit{query weight}
is the squared norm of the amplitudes of the states that $O_f$ acts on.
The sum of these weights is the Las Vegas query complexity of $B$ on input $f$.
Similarly, the query complexity of $B$ is the maximum query complexity over all inputs $f\in D$,
and the Las Vegas query complexity of $\phi$ is the minimum query complexity of all algorithms.
We let $\ql(\phi,\calO)$ denote the Las Vegas query complexity of a problem $\phi$
with oracle set $\calO$,
and in particular let
$\qlxor(\phi)$ and $\qlperm(\phi)$
denote the Las Vegas quantum query complexity of $\phi$ when using \XOR{} oracles and in-place oracles, respectively.

The techniques developed by \cite{Belovs2023onewayticketlasvegas} are more powerful than we need here,
as they primarily consider state-conversion problems,
where an algorithm transforms some starting state
$\xi_f$ to an ending state $\tau_f$.
Since we only consider decision problems,
our discussion of their adversary optimization problem \cite[Definition 7.3]{Belovs2023onewayticketlasvegas}
can be restricted to
$\xi_f=\ket{0}$ and $\tau_f=\ket{\phi(f)}$ for all $f\in D$.
This significantly simplifies our analysis.

In the classical setting,
a Las Vegas algorithm can be turned into a bounded-error algorithm with a similar complexity.
Given the Las Vegas complexity is $L$,
the algorithm can just be terminated after, say, $3L$ queries
to create a bounded-error algorithm which, by Markov's inequality,
correctly terminates with probability at least $2/3$ and errs with probability at most $1/3$.
Therefore, up to constant factors,
the bounded-error complexity is at most the Las Vegas complexity.
In the quantum setting,
\cite[Corollary 7.7]{Belovs2023onewayticketlasvegas}
remarks that the same relationship holds, and so we have
\begin{equation}\label{eqn:BEleLV}
    \qbxor(\phi) = \bO{\qlxor(\phi)}, \quad\text{and}\quad \qbperm(\phi) = \bO{\qlperm(\phi)}.
\end{equation}

\subsection{The Quantum Adversary Bound}
Quantum query complexity can be characterized by the adversary method.
This method has been used to develop several different adversary bounds or adversary theorems
in different contexts.
For example, prior work derived adversary bounds in the \XOR{} oracle model.
In general, an adversary bound for a decision problem $\phi:D\to \{0,1\}$
is an optimization problem such that
the optimum is a lower bound on the query complexity.
The version of the adversary bound developed by \cite{Belovs2023onewayticketlasvegas}
applies to arbitrary linear transformations.
In fact, \cite[Section 10]{Belovs2023onewayticketlasvegas} observed this includes in-place oracles
as well as \XOR{} oracles.
Crucially, their method is able to analyze \textit{unidirectional access} to an oracle $O_f$,
whereas other methods assume access to an oracle and its inverse
or assume an oracle is self-inverse.

To state the adversary bound of \cite{Belovs2023onewayticketlasvegas},
we must introduce the
``unidirectional relative ${\g}$ function'', denoted $\relg$.
This is a variant of the $\g$ norm,
also known as the Schur or Hadamard product operator norm
(see the lecture notes by Ben-David \cite{Gamma2_Shalev} for a nice introduction).

\begin{definition}[The $\g$ Norm]\label{def:g}
    The $\g$ norm for a matrix $E\in \complex^{n \times m}$ is
    \[
        \g(E) = \max_{\substack{\Gamma\in \complex^{n\times m} \\ \Gamma\neq 0}}
        \frac{\norm{\Gamma\circ E}}{{\norm{\Gamma}}} .
    \]
\end{definition}

For two matrices of equal dimensions,
$\circ$ denotes the Hadamard (a.k.a.\ Schur or element-wise) product.
Intuitively, $\g(E)$ is the operator norm of the map
$\Gamma \to \Gamma\circ E$.

To define $\relg$, we first need to extend the definition of $\circ$.
Consider some matrix $\Gamma\in \complex^{n\times m}$
and some potentially larger matrix $\Delta\in \complex^{\ell n\times \ell m}$ for $\ell\in\nats^+$.
We let $\Gamma[f,g]$ denote the scalar entry in row $f$ and column $g$ of $\Gamma$.
For the larger matrix, we let $\Delta[f,g]$ denote the $\ell\times\ell$ block at position $(f,g)$ of $\Delta$. Then, we define $\Gamma\circ\Delta = \Delta\circ\Gamma \in \complex^{\ell n\times \ell m}$ by defining each $\ell \times \ell$ block
\[
    \paren{\Gamma\circ\Delta}[f,g]= \Gamma[f,g] \cdot \Delta[f,g] ,
\]
where again, $\Gamma[f,g]$ is a scalar and $\Delta[f,g]$ is an $\ell \times \ell$ block.

\begin{definition}
[The Unidirectional Relative ${\g}$ Function $\relg$ \cite{Belovs2023onewayticketlasvegas}]
    The relative ${\g}$ function $\relg$ for a Hermitian matrix $E\in\complex^{n\times n}$ relative to another Hermitian matrix $\Delta\in\complex^{\ell n\times \ell n}$ is
    \[
        \relg(E,\Delta)=\max_{\substack{\Gamma\in \complex^{n\times n} \\ \Gamma=\Gamma^\dagger\neq0}}
        \frac{\lmax\p{\Gamma\circ E}}{\lmax\p{\Gamma\circ \Delta}} .
    \]
    \label{def:relg}
\end{definition}

Note that $\relg$ is stated in terms of the maximum eigenvalue,
while most other adversary bounds are stated in terms of the spectral norm. Also note that the optimization is over Hermitian $\Gamma$.

The adversary bound of \cite{Belovs2023onewayticketlasvegas} can be stated in terms of two matrices given as input to $\relg$.
One depends on the decision problem to be solved, and the other depends on the oracle model defining how instances of the problem can be queried.
The \textit{problem matrix} for a problem $\phi:D\to\zo$ is the $|D| \times |D|$ symmetric matrix
$E_\phi$ defined by
\[
    E_\phi[f,g]=\begin{cases}
        1 &\text{if $\phi(f)\neq\phi(g)$}\\
        0&\text{otherwise.}
    \end{cases}
\]
In other words, $E_\phi$ is the indicator matrix for whether instances $f$ and $g$ have different answers.
For an oracle model $\calO$,
which defines the map from $f\in D$ to the oracle unitary $O_f$,
the \textit{oracle matrix} is the $|D| \times |D|$ symmetric block matrix $\Delta_{\mathcal O}$
defined by
\begin{equation}\label{eqn:DeltaO}
    \Delta_{\calO} [f,g] = I-O_f^\dagger O_g .
\end{equation}
In particular, we will use $\dxor$ and $\dperm$ to denote the oracle matrices
in the \XOR{} and in-place oracle models, respectively.

The following result characterizes Las Vegas query complexity in various oracle models
by the $\relg$ function of $E_\phi$ relative to the oracle matrix $\Delta_{\calO}$.

\begin{theorem}[\cite{Belovs2023onewayticketlasvegas}]\label{thm:LAadv}
    For a problem $\phi:D \to \zo$ and oracle set $\mathcal O=\cb{O_f \mid f\in D}$,
    \[
        \ql(\phi,\calO)=\relg(E_\phi, \Delta_{\mathcal O}).
    \]
\end{theorem}
Note that this statement of the query complexity is tight.
The optimum value of the optimization problem is both an upper and lower bound on the query complexity.

Later in this section,
it will be very convenient for all oracle matrices to have the same dimensions.
$\Spi$ is generically larger than $\Ppi$ because of the additional register.
Fortunately, query complexity with \XOR{} oracles is known to be equivalent to query complexity with
\textit{phase oracles}.
Here, a phase oracle is defined to map $\ket{x}$ to $\omega_N^{f(x)}\ket{x}$ where $\omega_N$ is the $N$-th root of unity
and $x\in [N]$.
Using this definition of a phase oracle,
we may define the oracle matrix $\dphase$.
As desired, the unitary for a phase oracle has the same dimensions as the unitary for an in-place oracle.
So below,
we make statements about query complexity with \XOR{} oracles, $\qbxor$,
in terms of the oracle matrix for phase oracles, $\dphase$.

\subsection{The Structure of Adversary Matrices}
In the definition of $\relg$,
feasible solutions $\Gamma$ to the optimization problem are generally called \textit{adversary matrices},
and the goal is to identify an adversary matrix for a given problem with high objective value, implying a strong query lower bound.

Because \cref{thm:LAadv} is tight,
and because of the relationship between \XOR{} and phase oracles,
we have the following corollary relevant to our goal of finding a separation between \XOR{}
and in-place oracles.

\begin{corollary}\label{cor:sep}
    For a problem $\phi:D\to \zo$,
    there is a separation
    $\qlxor(\phi) < \qlperm(\phi)$ if and only if
    $\relg(\phi,\dphase) < \relg(\phi,\dperm)$.
\end{corollary}

We will show that if
the optimization in $\relg$ is restricted to a certain subset of adversary matrices,
then it is impossible to demonstrate an advantage of \XOR{} oracles over in-place oracles.
In other versions of the adversary method,
adversary matrices have been restricted to nonnegative values (the positive weight method)
or to general real numbers (the negative weights method).
For a decision problem $\phi:D\to\zo$,
previous methods have nearly always restricted $\Gamma$ such that
an entry $\Gamma[f,g]$ indexed by problem instances $f$ and $g$
satisfies that if $\phi(f)=\phi(g)$, then $\Gamma[f,g]=0$.
But, this new version of the adversary method removes that restriction:
we are free to assign nonzero values to \textit{all} entries of $\Gamma$.
We show that, just as negative-weight adversary matrices are necessary to prove tight lower bounds
for certain problems,
these ``extended'' adversary matrices,
with nonzero entries corresponding to problem instances with the same answer,
are necessary to prove the desired decision-problem separation
with \XOR{} oracles outperforming in-place oracles.
In other words, if we use only tools from the negative-weight adversary bound to construct adversary matrices $\Gamma$,
then we cannot prove such a query separation.

\begin{definition}
    An adversary matrix $\Gamma$ is considered \textit{extended} if some entry
    $\Gamma[f,g]$ is nonzero with $\phi(f)=\phi(g)$.
\end{definition}

It is useful to observe that an adversary matrix $\Gamma$ is extended
if and only if $\Gamma \neq \Gamma \circ E_\phi$.
Let us also define modified versions of $\relg$ in terms of extended and non-extended adversary matrices.

\begin{definition}
    $\relgext$ and $\relgstd$ for a matrix $E_\phi\in\complex^{n\times n}$ relative to another matrix $\Delta\in\complex^{\ell n\times \ell n}$ are defined by
    \[
    \relgext(E_\phi,\Delta)=\max_{\substack{
        \Gamma\in \complex^{n\times n} \\ \Gamma=\Gamma^\dagger\neq 0 \\ \Gamma\neq\Gamma\circ E_\phi}
        }
    \frac{\lmax\p{\Gamma\circ E_\phi}}{\lmax\p{\Gamma\circ \Delta}}
    \quad\text{and}\quad
    \relgstd(E_\phi,\Delta)=\max_{\substack{
        \Gamma\in \complex^{n\times n} \\ \Gamma=\Gamma^\dagger\neq 0 \\ \Gamma=\Gamma\circ E_\phi}
        }
    \frac{\lmax\p{\Gamma\circ E_\phi}}{\lmax\p{\Gamma\circ \Delta}} .
    \]
\end{definition}

We show the following.

\begin{lemma}\label{thm:needextend}
    For any decision problem $\phi:D\to \zo$,
    $\relgstd(\phi,\dperm) \leq 3 \relgstd(\phi,\dphase)$.
\end{lemma}

The implication of \cref{thm:needextend} is that
non-extended adversary matrices always yield a weaker query lower bound for in-place oracles than
\XOR{} oracles oracles when applied to \cref{thm:LAadv}.
Because the adversary bound is tight, if a separation exists, it must be provable using \cref{thm:LAadv}.
Therefore,
there exists a decision problem separation in Las Vegas complexity with \XOR{} oracles outperforming
in-place oracles if and only if it is witnessed by extended adversary matrices.
Formally,
combining \cref{thm:needextend,cor:sep} yields the following theorem.

\begin{restatable}{theorem}{corneedextend}
\label{cor:needextend}
    For a decision problem $\phi:D\to \zo$,
    there is a separation
    $\qlxor(\phi) \leq \lO{\qlperm(\phi)}$ if and only if
    $\relgext(\phi,\dphase)\leq \lO{\relgext(\phi,\dperm)}$.
\end{restatable}

To prove \cref{thm:needextend}, we use the following two claims.
Recall $\g(X)$ (\cref{def:g})
is intuitively
the operator norm of the map $\Gamma\to\Gamma\circ X$.
\Cref{lem:schuropt} gives a method to bound $\g(X)$.
In \cref{lem:phasenorm},
we show that $\g(\dphase\circ E_\phi)\le 3$,
meaning element-wise multiplication by $\dphase\circ E_\phi$
cannot increase the norm of a matrix more than a factor of 3.

\begin{lemma}[Theorem 5.1 in \cite{pisier2004similarity}]
    \label{lem:schuropt}
    For any $A\in \mathbb C^{n\times m}$,
    ${\g}\p{A}\leq c$ if and only if there exists some $d\ge 1$ and vectors
    $u_i\in \mathbb C^d$ and $v_j\in \mathbb C^d$
    such that $A_{ij}=\braket{u_i| v_j}$
    for all $i\in [n]$ and $j\in [m]$ and
    \begin{align*}
        &\max_{i\in [n],~j\in [m]}\norm{u_i}\norm{v_j}\leq c.
    \end{align*}
\end{lemma}

\begin{proposition}\label{lem:phasenorm}
  $\g\p{\dphase\circ E_\phi} \leq 3$.
\end{proposition}

To prove \cref{lem:phasenorm}, we first decompose $\dphase\circ E_\phi$ as $\dphase\circ E_\phi=\Delta' - \Delta''$. For all $f,g$ satisfying $\phi(f)\not= \phi(g)$, we define $\Delta'[f,g]=I$ and $\Delta''[f,g]=O_f^\dagger O_g$ (all other blocks are zero). We will show that $\g(\Delta') \le 1$ and $\g(\Delta'')\le 2$ using \cref{lem:schuropt}, and then apply the triangle inequality to get \cref{lem:phasenorm}.

\begin{lemma}
    $\g(\Delta')\le1$\label{Deltap}
\end{lemma}
\begin{proof}
    We have
    \[
        \Delta'=\sum_{j\in[N]}
        \sum_{\substack{f,g \\ \phi(f)\not=\phi(g)}}
        \ketbra{f,j}{g,j}.
    \]
    For all $f\in D$ and $j\in [N]$, let
    $\ket{u_{f,j}}=\ket{j}\ket{\phi(f)}$ and
    $\ket{v_{f,j}}=\ket j\ket{\phi(f)\oplus 1}$.
    Then \[\ang{u_{f,j}|v_{g,i}}=\begin{cases}
        1&\text{if $i=j$ and $\phi(f)\not=\phi(g)$}\\
        0&\text{otherwise,}
    \end{cases}\]
    so $\bra{f,j}\Delta'\ket{g,i}=\ang{u_{f,j}|v_{g,i}}$ for all $f,g\in D$ and $i,j\in [N]$.
    By \cref{lem:schuropt}, $\g(\Delta')\le 1$.
\end{proof}

The proof for $\Delta''$ follows similarly.
\begin{lemma}
    $\g\p{\Delta''}\le 2$\label{Deltapp}
\end{lemma}

\begin{proof}
    We have
    \[
        \Delta'' = \sum_{j\in[N]}
        \sum_{\substack{f,g \\ \phi(f)\neq\phi(g) \\ f(j)\neq g(j)}}
        \omega_N^{-f(j)+g(j)}
        \ketbra{f,j}{g,j}.
    \]
    For all $f\in D$ and $j\in [N]$, let
    \[
         \ket{u_{f,j}} = \omega_N^{f(j)}\ket{j}\ket{\phi(f)}\p{\ket{N+1} + \ket{f(j)}}
    \]
    and
    \[
         \ket{v_{f,j}} = \omega_N^{f(j)}\ket{j}\ket{\phi(f)\oplus 1}\p{\ket{N+1} - \ket{f(j)}}
    \]
    If either $\phi(f)=\phi(g)$ or $i\neq j$, then $\braket{u_{f,j}\mid v_{g,i}}=0$.
    Otherwise,
    \[
        \braket{u_{f,j}\mid v_{g,j}}
        = \omega_N^{-f(j)+g(j)}(\bra{N+1} +\bra{f(j)})(\ket{N+1} -\ket{g(j)})
        = \omega_N^{-f(j)+g(j)}(1-\braket{f(j)|g(j)})
    \]
    which is equal to $\omega_N^{-f(j)+g(j)}$ if $f(j)\neq g(j)$ and zero otherwise.

    Thus, $\bra{f,j}\Delta' \ket{g,i} = \braket{u_{f,j}|v_{g,i}}$ for all $f,g\in D$ and $i,j\in [N]$.
    By \cref{lem:schuropt}, $\g(\Delta'') \le 2$.
\end{proof}

To get \cref{lem:phasenorm}, we just need to apply the triangle inequality to the results from \cref{Deltap} and \cref{Deltapp}.
\begin{proof}[Proof of \cref{lem:phasenorm}]
    We have \begin{align*}
        \g(\dphase\circ E_\phi)&=\g(\Delta'-\Delta'')\\
        &\le \g(\Delta') + \g(\Delta'') &&\text{Triangle inequality}\\
        &\le 3&&\text{\cref{Deltap,Deltapp}}.
    \end{align*}
\end{proof}
Finally, we prove \cref{thm:needextend}.

\begin{proof}[Proof of \cref{thm:needextend}]
    Notice that each block $\dphase[f,g]$ is diagonal.
    If an entry of $\dphase$ is nonzero,
    then the same entry of $\dperm$ is $1$.
    In other words, $\dphase = \dphase\circ \dperm$.
    So, if our adversary matrix $\Gamma$ is only defined on entries $(f,g)$ in which $\phi(f)\neq\phi(g)$, then
    \begin{align*}
        \norm{\Gamma\circ \dphase}&=\norm{\Gamma\circ E_\phi\circ \dphase}&&\text{since $\Gamma=\Gamma\circ E_\phi$}\\
        &=\norm{\Gamma\circ \dperm\circ \p{E_\phi\circ \dphase}}&&\text{since $\dphase=\dperm\circ\dphase$}\\
        &\leq 3\norm{\Gamma \circ\dperm}&&\text{\cref{lem:phasenorm}} .
    \end{align*}

    Now, note for any decision problem $\phi$, that $E_\phi$ is the adjacency matrix for an unweighted undirected bipartite graph. This means that for symmetric $\Delta$ and non-extended $\Gamma$, the Hermitian matrix $\Gamma \circ \Delta = E_\phi \circ \Gamma \circ \Delta$ is the adjacency matrix for a weighted bipartite graph. For such matrices, the maximum eigenvalue and operator norm are equal, so we see
    \[
    \lmax{(\Gamma\circ \dphase)} = \norm{\Gamma\circ \dphase} \le 3\norm{\Gamma \circ\dperm} = 3\lmax{(\Gamma \circ\dperm)}.
    \]
    Hermitian matrices with trace of zero have non-negative maximum eigenvalues, so for any non-extended adversary matrix $\Gamma$,
    \[
    \frac{\lmax\p{\Gamma \circ E_\phi}}{\lmax\p{\Gamma \circ \dperm}} \le 3\frac{\lmax\p{\Gamma \circ E_\phi}}{\lmax\p{\Gamma \circ \dphase}}.
    \]
    In particular, this holds for the non-extended adversary matrix optimizing $\relgstd(\phi,\dperm)$, so
    \[
    \relgstd(\phi,\dperm) \le 3 \relgstd(\phi,\dphase),
    \]
    concluding our proof.
\end{proof}

\subsection{A Candidate Decision-Problem Separation}
Now we apply our bound in the special case of $\Phi=\perminvgarb{}$ (\cref{def:permInvGarb}),
and we are able to extend the above statements about Las Vegas complexity
to standard bounded-error complexity.
With query access to an \XOR{} oracle,
it is straightforward to use Grover's search algorithm
for a bounded-error query complexity of $\qbxor(\Phi)=\bTnoparen(\sqrt{N})$ (\cref{lem:perminvgarbXor}).
Moreover, for XOR oracles, one can implement a zero-error version of Grover's search \cite{Hoyer2000-exactGrover}, 
which satisfies the conditions for a Las Vegas query algorithm. This means that $\qlxor(\Phi)=\bTnoparen(\sqrt{N})$.
On the other hand,
it is unclear how an in-place query algorithm might deal with the query-related garbage and
achieve any quantum speedup. 
Therefore, we conjecture that the in-place query complexity of this decision problem is $\bOm{N}$
(\cref{conj:perminvgarb}).

Any Las Vegas algorithm can be turned into a Monte Carlo algorithm with constant overhead (see \cref{eqn:BEleLV}). 
Therefore, $\qbperm(\Phi) \leq \bO{\qlperm(\Phi)}$.

Combining these two observations with \cref{cor:needextend}
yields the following theorem.

\begin{restatable}{theorem}{claimperminvgarbExtended}
\label{claim:perminvgarbExtended}
    For the decision problem $\Phi = \perminvgarb{}$,
    there is a separation
    $\qbxor(\Phi) \leq \lO{\qbperm(\phi)}$,
    if and only if
    $\relgext(\Phi,\dphase) \leq \lO{\relgext(\Phi,\dperm)}$.
\end{restatable}

So, we find that extended adversary matrices are crucial for
proving \cref{conj:perminvgarb}.

\end{document}